\documentclass[journal]{IEEEtran}
\usepackage{amsmath,amsfonts,amssymb}
\usepackage{algorithm}
\usepackage[noend]{algpseudocode}
\usepackage{array}
\usepackage[caption=false,font=normalsize,labelfont=sf,textfont=sf]{subfig}
\usepackage{textcomp}
\usepackage{stfloats}
\usepackage{url}
\usepackage{verbatim}
\usepackage{graphicx}
\usepackage{cite}
\usepackage{bm, bbm, xcolor}
\newtheorem{theorem}{Theorem}

\newtheorem{remark}{Remark}

\begin{document}

\title{Discrete Coupling and Localized Motion for Pinching-Antenna Systems (PASS)}

\author{Jie~Jiang,~\IEEEmembership{Graduate Student Member,~IEEE}, Xiaoxia Xu,~\IEEEmembership{Member,~IEEE}, Chan-Tong Lam, ~\IEEEmembership{Senior Member,~IEEE}, Yuanwei Liu,~\IEEEmembership{Fellow,~IEEE}, and Arumugam~Nallanathan,~\IEEEmembership{Fellow,~IEEE}
\thanks{Jie Jiang and Xiaoxia Xu are with the School of Electronic Engineering and Computer Science, Queen Mary University of London, London, U.K. (e-mails: {jie.jiang, x.xiaoxia}@qmul.ac.uk).}
\thanks{Chan-Tong Lam is with the Faculty of Applied Sciences, Macao Polytechnic University, Macao (e-mail: ctlam@mpu.edu.mo).}
\thanks{Yuanwei Liu is with the Department of Electrical and Electronic Engineering, The University of Hong Kong, Hong Kong (e-mail: yuanwei@hku.hk).}
\thanks{Arumugam Nallanathan is with the School of Electronic Engineering and Computer Science, Queen Mary University of London, E1 4NS London, U.K., and also with the Department of Electronic Engineering, Kyung Hee University, Yongin-si, Gyeonggido 17104, Korea (e-mail: a.nallanathan@qmul.ac.uk).}
}



\maketitle

\begin{abstract}
The practical implementation of pinching-antenna systems (PASS) is challenging due to hardware limitations in large-scale antenna movement and continuous radiation power adjustment. 
This paper proposes a practical PASS-enabled downlink multi-user multiple-input multiple-output communication framework that enables discrete radiation power control and localized discrete antenna movement. Specifically, a discrete coupling strength model is exploited to tune the radiation power at each pinching antenna (PA) through quantized coupling spacing levels. Moreover, each PA can only move among discrete locations within a limited region determined by the movement speed and duration. Based on the proposed framework, a joint optimization problem of the PA positions, coupling strength, and transmit beamforming is formulated. Considering waveguide attenuation, the total average power consumption is minimized, subject to each user's minimum SINR requirement and localized motion constraints. To address this coupled mixed-integer nonconvex optimization problem, a globally optimal branch-and-bound-based algorithm is first developed for the multi-waveguide single-user scenario. To further reduce complexity, a scalable genetic algorithm-assisted particle swarm optimization (GA-PSO) method is developed for the multi-waveguide multi-user scenario, where GA operations are incorporated to preserve population diversity and alleviate premature convergence. Simulation results demonstrate that the proposed design significantly reduces the power consumption compared with the conventional PASS schemes and MIMO architectures.
\end{abstract}


\section{Introduction}
Future sixth-generation (6G) wireless networks are expected to support ultra-high data rates, massive connectivity, and reliable services in highly dynamic propagation environments \cite{you20216G}. To meet these requirements, flexible-antenna technologies have recently attracted increasing attention, since they enable proactive reconfiguration of wireless channels and provide additional spatial degrees of freedom (DoFs) \cite{new2025Fluid}. Representative flexible-antenna architectures include reconfigurable intelligent surfaces (RISs) \cite{liuReconfigurableIntelligentSurfaces2021}, fluid antennas \cite{wongFluidAntennaSystems2020}, and movable antennas \cite{zhu2025Tutoriala}. By adjusting electromagnetic responses or antenna positions, these technologies can improve channel conditions and enhance communication performance \cite{new2025Fluid}. 
To further mitigate the large-scale path loss and maintain stable line-of-sight (LoS) links, pinching-antenna systems (PASS) have recently emerged as a promising flexible-antenna architecture for 6G wireless communications \cite{dingFlexibleantennaSystemsPinchingantenna2025}. PASS employ dielectric waveguides as the primary transmission medium, where small dielectric particles, referred to as pinching antennas (PAs), are attached to the waveguides to radiate guided signals into free space \cite{liuPinchingantennaSystemsPASS2026}. Since dielectric waveguides can extend tens of meters with low propagation loss, PAs can be deployed close to users to establish strong LoS links and reduce free-space path loss. As a result, PASS provide several unique advantages, including large-scale antenna reconfiguration, last-meter communication capability, scalable deployment, and low-cost implementation \cite{xu2026Generalized}.

Motivated by these advantages, extensive research efforts have been devoted to PASS-enabled wireless communications. \textit{Pinching beamforming} has been introduced as a key PASS technique, where PA positions are optimized to reconfigure both large-scale path loss and signal phases \cite{dingFlexibleantennaSystemsPinchingantenna2025,wangModelingBeamformingOptimization2025,xuJointTransmitPinching2026,wang2025Joint,qin2025Joint,xuRateMaximizationDownlink2025,liuPinchingantennaSystemsPASS2026,xu2026Generalized,sunmultiuserBeamformingPinchingantenna2025,wang2025Antenna}. Specifically, the PA positions can be generally tuned by two structures, namely continuous activation and discrete activation \cite{liuPinchingantennaSystemsPASS2026}. Continuous activation assumes that PA positions can be continuously adjusted along the waveguide, thereby providing high spatial flexibility for pinching beamforming \cite{dingFlexibleantennaSystemsPinchingantenna2025,xuJointTransmitPinching2026,wang2025Joint,qin2025Joint,xuRateMaximizationDownlink2025}. The authors in \cite{dingFlexibleantennaSystemsPinchingantenna2025} provided one of the first analytical studies of PASS, showing that the meter-scale spatial reconfiguration of PAs along the waveguide can effectively reduce large-scale path loss, enhance average user rates, and enable flexible extensions to non-orthogonal multiple access (NOMA) and multi-waveguide multiple-input single-output (MISO) transmissions. The authors in \cite{xuJointTransmitPinching2026} further generalized PASS to a downlink multi-user MISO architecture with multiple PAs on each waveguide, where PA locations were continuously optimized to jointly reshape path loss and signal phases.
In contrast, discrete-activation restricts PA activation to a finite set of predefined candidate positions, which is more suitable for low-cost practical implementation \cite{tyrovolasErgodicRateAnalysis2025,wangModelingBeamformingOptimization2025,wangAntennaActivationResource2026}. Bridging the physical modeling and beamforming design of PASS, the authors in \cite{wangModelingBeamformingOptimization2025} developed a physics-based hardware model and jointly optimized transmit beamforming and pinching beamforming under both continuous and discrete PA activation schemes. Following this line, the authors in \cite{wangAntennaActivationResource2026} investigated a practical multi-waveguide PASS with pre-configured discrete PA positions, where waveguide assignment, antenna activation, SIC decoding order, and power allocation were jointly optimized to maximize the sum rate in NOMA-assisted downlink transmission.

Beyond the spatial reconfiguration of PAs, coupling-induced radiation power distribution among multiple PAs is another key factor affecting PASS performance \cite{wangModelingBeamformingOptimization2025,xuPinchingantennaSystemsPASS2025}.
Different from conventional antenna arrays where the radiated power of each antenna element can be directly controlled through dedicated RF chains, the radiation power distribution among PAs is determined by the electromagnetic coupling process between dielectric waveguides and PAs \cite{wangModelingBeamformingOptimization2025,haus1991Coupledmode,millerCoupledWaveTheory1954,marcuseTheoryDielectricOptical2013}. To characterize this physical process, the coupled-mode theory (CMT) provides a fundamental basis for modeling the coupling strength and the resulting radiation behavior \cite{haus1991Coupledmode}. In \cite{wangModelingBeamformingOptimization2025}, each PA was modeled as an open-ended directional coupler, based on which equal-power and proportional-power radiation models were developed for multiple PAs along the same waveguide. To enable controllable radiation power allocation beyond fixed radiation assumptions, the authors in \cite{xuPinchingantennaSystemsPASS2025} proposed an adjustable power radiation model, where the radiation power ratios of PAs are controlled by tuning the coupling spacing between PAs and waveguides. Closed-form spacing expressions were derived for equal-power radiation, and a discrete activation PASS beamforming problem was formulated for transmit power minimization.

Several recent studies have further moved toward practical PASS modeling \cite{xuJointRadiationPower2025,xuPinchingantennaSystemsInwaveguide2026,huSumRateMaximizationPinching2025b,zhangDirectionalPinchingAntennaSystems2025}. In \cite{xuJointRadiationPower2025}, motion power consumption was incorporated into PASS design, where PA positions, radiation power ratios, and transmit beamforming were jointly optimized under both continuous and discrete antenna movement. In \cite{xuPinchingantennaSystemsInwaveguide2026}, in-waveguide attenuation was explicitly incorporated into system modeling and algorithm design, and its impact on the optimal PA placement and achievable rate was analyzed. In addition, frequency-dependent waveguide attenuation was considered for PA-assisted NOMA systems with multiple dielectric waveguides in \cite{huSumRateMaximizationPinching2025b}, showing that waveguide attenuation may significantly affect system performance, especially at high carrier frequencies.

Despite these advancements, existing PASS designs still rely on idealized assumptions of PA motions and power radiation. On the one hand, most beamforming-oriented studies assume unconstrained antenna position optimization or continuous activation over the entire waveguide \cite{dingFlexibleantennaSystemsPinchingantenna2025,xuJointTransmitPinching2026,wang2025Joint,qin2025Joint}, whereas practical PA movement is constrained by actuator speed, motion latency, motion energy consumption, and mechanical control resolution \cite{basbug2017Design}. In practical implementations, PAs can usually be adjusted only within a localized region and over a finite set of candidate positions during each transmission frame. On the other hand, existing power radiation models are commonly based on equal-power radiation or continuously adjustable radiation control \cite{xuBeamformingPinchingAntenna2025,wangModelingBeamformingOptimization2025,xuPinchingantennaSystemsPASS2025,xuJointRadiationPower2025}. Although coupling-length- and spacing-based models provide physics-grounded ways to configure radiation ratios, they usually require either fabricated coupling structures or continuously and accurately controllable coupling spacing, which is difficult to realize in practical electromechanical or RF switching implementations. As a result, coupling strength control with finite discrete coupling spacing states tailored to practical hardware remains insufficiently investigated. Furthermore, although in-waveguide attenuation has been separately studied, its impact has not been fully integrated into PASS designs with localized PA movement and discrete radiation control. Therefore, how to unify these hardware-induced constraints into a tractable joint optimization framework remains underexplored.

Motivated by the above observations, this paper investigates a practical downlink multi-user PASS framework by jointly considering localized PA movement, discrete coupling strength, and in-waveguide attenuation. Specifically, each PA can only move within a limited local region and select its position from a finite set of feasible mounting points. A hardware-compatible discrete coupling strength model is adopted, where each PA selects one coupling spacing level from a finite set to control its radiation power ratio. Practical in-waveguide attenuation and phase rotation are also incorporated into the cascaded PASS channel model. Under this framework, the PA positions, coupling strength, and transmit beamforming are jointly optimized to minimize the total power consumption, including both transmit power and PA motion power, subject to users' SINR requirements and localized motion constraints. The formulated problem is a highly coupled mixed-integer nonconvex problem, since the discrete PA positions simultaneously affect the free-space path loss, propagation phase, and in-waveguide attenuation, while the discrete coupling strength determines the cascaded radiation coefficients and reshapes the power distribution among activated PAs. As a result, the discrete structural variables and continuous transmit beamforming are tightly intertwined in both the objective function and the SINR constraints. To address this challenge, we develop a globally optimal branch-and-bound (BnB) algorithm for the multi-waveguide single-user (MWSU) scenario and a scalable genetic algorithm (GA)-assisted particle swarm optimization (GA-PSO) algorithm for the general multi-waveguide multi-user (MWMU) scenario.

The main contributions of this paper are summarized as follows.
\begin{itemize}
\item
We propose a practical PASS framework that enables discrete coupling strength control and localized discrete PA movement. We first conceive a discrete coupling strength model, which places PAs at predefined coupling spacing levels to adjust discrete power radiation. Moreover, the position of each PA can only be selected from a set of mounting points within a localized region determined by the PA movement speed and movement duration. Considering in-waveguide attenuation, the total power consumption for both transmission and PA motion is minimized under users' minimum SINR requirements and PAs' motion range constraints. This is formulated as a joint optimization problem of the localized PA positions, discrete coupling strength, and transmit beamforming.

\item
To address the resulting highly coupled mixed-integer nonconvex problem, we first develop a globally optimal BnB algorithm for the MWSU scenario. By exploiting the maximum-ratio transmission (MRT) structure, the continuous transmit beamforming variable is eliminated in closed form, and the original problem is reduced to a discrete structural optimization problem. The proposed BnB algorithm is proved to converge to the global optimum within a prescribed optimality tolerance.

\item
For the MWMU scenario, we further develop a scalable GA-PSO algorithm to reduce the complexity for the nonconvex problem. The proposed algorithm adopts a dual-layer structure. The outer-layer GA-assisted PSO explores the discrete PA positions and coupling spacing levels, while the inner-layer second-order cone programming (SOCP) optimally solves the transmit beamforming subproblem for each candidate PASS configuration. The GA module is effectively introduced to enhance population diversity and alleviate premature convergence in the discrete search space.

\item
Simulation results verify the effectiveness of the proposed practical PASS framework and the developed algorithms. In particular, the comparison with the equal-power discrete-motion PASS baseline shows that coupling strength control effectively avoids inefficient uniform radiation, while the comparison with the adjustable-coupling discrete-activation PASS baseline confirms that localized PA movement is critical for reshaping the effective channel gains. Moreover, the proposed PASS consistently outperforms conventional multiple-input multiple-output (MIMO) and hybrid MIMO architectures, highlighting its capability in mitigating large-scale path loss and enhancing spatial beamforming flexibility through reconfigurable pinching beamforming.
\end{itemize}

The rest of this paper is organized as follows. 
Section II presents the proposed PASS framework with discrete PA movement and an adjustable power radiation model, and formulates the corresponding optimization problem. 
Section III proposes a globally optimal BnB algorithm for the MWSU scenario, while Section IV develops the GA-PSO algorithm for the MWMU scenario. 
Section V provides numerical results to verify the effectiveness of the proposed framework and algorithms. 
Finally, Section VI concludes the paper.

\textit{Notations}: The variable, vector, and matrix are denoted by $x$, $\mathbf{x}$, and $\mathbf{X}$, respectively. $|x|$ denotes the absolute value of a real number and the modulus of a complex number. $\|\mathbf{x}\|$ denotes the Euclidean norm.
$\Re\left\{x\right\}$ and $\Im\left\{x\right\}$ denote the real and imaginary parts of $x$. 
$\mathbf{X}^{T}$ and $\mathbf{X}^{H}$ denote the transpose and Hermitian transpose of $\mathbf{X}$, respectively.

\section{System model}

\begin{figure}[t]
  \centering
  \includegraphics[width=0.45\textwidth]{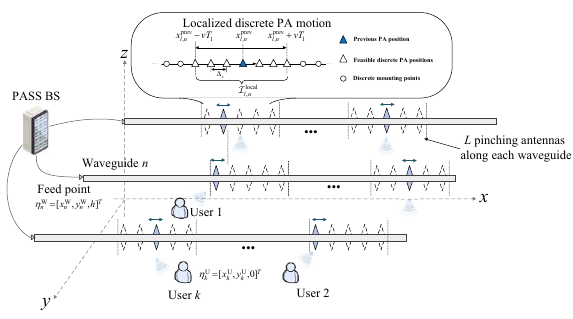}%systemmodel.eps}
  \caption{\label{fig:model}Illustration of the proposed PASS-enabled downlink MIMO communication system.
  }
\end{figure}

\begin{figure}[t]
  \centering
  \includegraphics[width=0.45\textwidth]{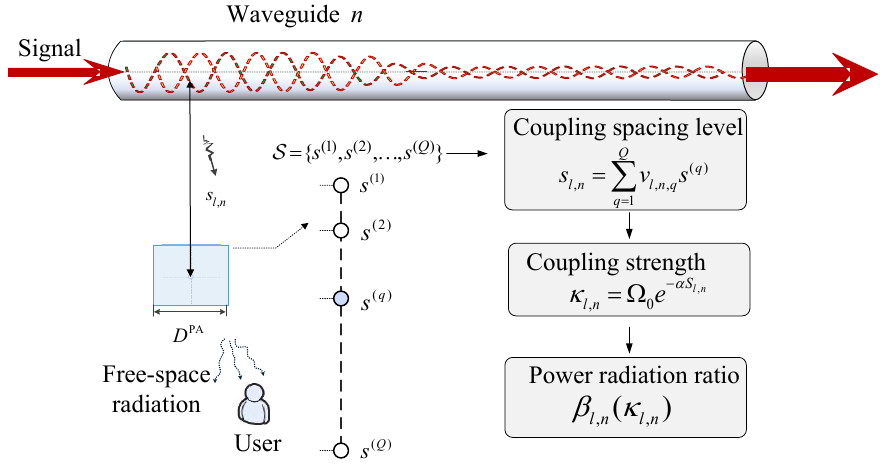}%systemmodel.eps}
  \caption{\label{fig:spacing}Discrete coupling strength control for adjustable power radiation. }
\end{figure}

As illustrated in Fig.~\ref{fig:model}, we consider a PASS-enabled downlink multi-user multiple-input multiple-output network, where a multi-antenna base station (BS) serves a set of single-antenna users indexed by $\mathcal{K}=\{1,\dots,K\}$. 
The BS is equipped with $N$ dielectric waveguides, indexed by $\mathcal{N}=\{1,\dots,N\}$. Each waveguide is connected to an independent RF chain and can radiate signals through multiple PAs. Along each waveguide, $L$ PAs, indexed by $\mathcal{L}=\{1,\dots,L\}$, are deployed, whose longitudinal positions can be dynamically adjusted within a set of discrete feasible locations. The total number of PAs in the system is therefore $M=NL$. For ease of notation, we use the tuple $(l,n)$ to denote the $l$-th PA on the $n$-th waveguide.

\subsection{Coupling Strength Control for Adjustable Power Radiation}

The power radiation ratio of each PA can be flexibly regulated by controlling the coupling strength between the waveguide and PAs. Specifically, the coupling strength refers to the local PA-waveguide power exchange strength, which is characterized by the coupling coefficient $\kappa_{l,n}$. For a fixed fabricated coupling length $D^{\mathrm{PA}}$, it can be equivalently reflected by the local radiation factor $\sin(\kappa_{l,n}D^{\mathrm{PA}})$. In practice, the local coupling strength can be tuned by varying the coupling spacing. The coupling spacing between the $(l,n)$-th PA and the waveguide $s_{l,n}$ is defined as the perpendicular distance between the waveguide surface and the radiating element of the PA \cite{liuPinchingantennaSystemsPASS2026}. Modifying this spacing changes how deeply the PA interacts with the evanescent field surrounding the waveguide, thereby adjusting the amount of guided power coupled into the PA and radiated into free space.

Following the CMT \cite{haus1991Coupledmode,wangModelingBeamformingOptimization2025}, the coupling coefficient $\kappa_{l,n}$ between waveguide $n$ and the $(l,n)$-th PA is characterized by the overlap integral of their modal fields, i.e.,
\begin{equation}\label{eq:CMT_integral_discrete}
    \kappa_{l,n}
    = \frac{\omega\epsilon_{0}}{4}
      \iiint_{V(s)}
      \Delta\epsilon \,
      \mathbf{E}_{{\rm wg},n}\!\cdot\!
      \mathbf{E}_{{\rm pa},l}^{*}
      \,\mathrm{d}V(s),
\end{equation}
where $\mathbf{E}_{{\rm wg},n}$ and $\mathbf{E}_{{\rm pa},l}$ are the power-normalized electric field distributions of waveguide $n$ and PA $l$, respectively, $\epsilon_{0}$ is the vacuum permittivity, $\Delta\epsilon$ denotes the dielectric perturbation associated with the pinching structure,  $\omega$ is the angular frequency, and $V(s)$ is the coupling region determined by $s$.

Following the analytical approximation in \cite{xuPinchingantennaSystemsPASS2025}, $\kappa_{l,n}$ can be represented as an exponentially decaying function of the spacing, which is given by
\begin{equation}\label{eq:coupling_coefficient}
    \kappa_{l,n}(s_{l,n}) = \Omega_{0} e^{-\alpha s_{l,n}}, \, s_{l,n} \in [s_{\min}, s_{\max}],\, \forall l, n,
\end{equation}
where $\Omega_{0}$ captures mode-overlap and normalization effects, and  $\alpha = \sqrt{\gamma_{0}^{2} -\frac{4\pi^2}{\lambda_{f}^2}n_{\mathrm{clad}}^{2}}$ is the cladding decay constant determined by the guided-mode propagation constant and the cladding refractive index. This relationship explicitly captures the diminishing nature of evanescent coupling as the separation increases, providing a practical mechanism to tune the coupling strength via spacing adjustments.

Although prior PASS power radiation models have shown that equal or flexible radiation can be achieved by customizing the coupling length of each PA, such a design is not well suited for real-time reconfiguration, since the coupling length is typically fixed after fabrication \cite{wangModelingBeamformingOptimization2025}. Alternatively, equal-power radiation can be realized by adjusting the coupling spacing. However, this approach requires continuous and highly accurate spacing control, which is difficult to guarantee in practical hardware \cite{xuPinchingantennaSystemsPASS2025}. In realistic implementations, the coupling spacing is usually controlled by finite-resolution mechanisms, such as electromagnetic actuation or RF MEMS switching, which can only provide a limited number of stable coupling states. As shown in Fig.~\ref{fig:spacing}, we adopt a discrete coupling spacing model, where the coupling spacing is quantized into a finite set of predefined spacing levels, denoted by $ \mathcal S=\{s^{(1)},s^{(2)},\ldots,s^{(Q)}\}$, where $Q$ is the number of available spacing levels. A uniform quantization scheme is adopted, i.e., $s^{(q)} = s_{\min}+(q-1)\Delta_s$, where $\Delta_{s} = (s_{\max}-s_{\min})/(Q-1)$. Let $v_{l,n,q}\in\{0,1\}$, $\forall q\in \mathcal{Q} = \{1,\dots,Q\}$, denote whether the $(l,n)$-th PA selects coupling spacing level $s^{(q)}$. Each PA selects only one coupling spacing level, satisfying $\sum_{q=1}^{Q}v_{l,n,q} = 1, \forall l,n$. Then, the coupling spacing is given by $ s_{l,n}=\sum_{q=1}^{Q}v_{l,n,q}s^{(q)}$.
Accordingly, the corresponding coupling strength between the waveguide $n$ and the $(l,n)$-th PA is given by
\begin{equation}\label{eq:kappa_discrete}
    \kappa_{l,n} 
    = \sum_{q=1}^{Q} v_{l,n,q} \Omega_{0} e^{-\alpha s^{(q)}}.
\end{equation}

Given the discrete coupling strength, the power radiation ratio of $(l,n)$-th PA is obtained based on the cascaded CMT-based model \cite{xuPinchingantennaSystemsPASS2025}. Since a PA is activated once it is installed at a selected mounting point, the power radiation ratio is expressed as
\begin{equation}\label{eq:beta_discrete}
    \beta_{l,n} = \sin\big(\kappa_{l,n} D^{\mathrm{PA}}\big)
      \prod_{i=1}^{l-1} \sqrt{1 -\sin^{2}\big(\kappa_{i,n} D^{\mathrm{PA}}\big)},
\end{equation}
where $D^{\mathrm{PA}}$ is the fixed fabricated coupling length of each PA. The selected spacing level $s_{l,n}$ determines the coupling strength $\kappa_{l,n}$ and thus the local radiation factor $\sin(\kappa_{l,n}D^{\mathrm{PA}})$, whereas the final radiation ratio $\beta_{l,n}$ is additionally determined by the residual guided power after the preceding PAs on the same waveguide.

\vspace{-1.5em}
\subsection{Discrete PA Movement Model}

We consider a three-dimensional Cartesian coordinate system. 
The BS is located at $\boldsymbol{\eta}_{0}=(0,0,h)$, where $h$ denotes the fixed deployment height. 
The dielectric waveguides are deployed parallel to the $x$-axis at height $h$. 
The feed point of waveguide $n \in \mathcal{N}$ is denoted by $\boldsymbol{\eta}_{n}^{\mathrm{W}}= [x_{n}^{\mathrm{W}},\,y_{n}^{\mathrm{W}},\,h]^{T}$, where $x_{n}^{\mathrm{W}}\in[0,D^{\mathrm{W}}]$ and $D^{\mathrm{W}}$ represents the waveguide length. 
We assume users remain quasi-static within each frame, and the position of user $k$ is given by $\boldsymbol{\eta}_{k}^{\mathrm{U}}= [x_{k}^{\mathrm{U}},\,y_{k}^{\mathrm{U}},\,0]^{T},\, \forall k\in\mathcal{K}$.

We adopt a two-phase protocol within each time frame of duration $T=T_1+T_2$ \cite{xuJointRadiationPower2025}.
During the first phase $T_1$, the PAs are repositioned along the waveguides, while in the second phase $T_2$, downlink transmission is performed using optimized beamforming and power radiation configurations.
Due to mechanical limitations, each PA can only move within a bounded range during $T_1$.
Let $\boldsymbol{\eta}_{l,n}^{\mathrm{prev}}=[x_{l,n}^{\mathrm{prev}},\,y_n^{\mathrm{W}},\,h]^T$ denote the position of the $(l,n)$-th PA at the beginning of the current frame, which corresponds to its final position in the previous frame. The movement distance is defined as $d_{l,n}^{\mathrm{mov}} = \big| x_{l,n} - x_{l,n}^{\mathrm{prev}} \big|$.
To accommodate the discrete movement capability of practical actuators, each waveguide is equipped with a predefined set of mounting points, denoted by 
$\mathcal{X}^{\mathrm{global}}=\{x^{(1)},x^{(2)},\dots,x^{(I)}\}\subset[0,D^{\mathrm{W}}]$,
where the mounting points are uniformly spaced with spatial resolution $\Delta_x$, i.e., $x^{(i+1)}-x^{(i)}=\Delta_x$, $\forall i=1,\ldots,I-1$. Here, $D^{\mathrm W}$ denotes the waveguide length, and $I$ is the number of mounting points on each waveguide. The corresponding global index set is defined as $\mathcal{I}^{\mathrm{global}}=\{1,2,\ldots,I\}$.

Unlike conventional continuous antenna movement over the entire waveguide, practical PA movement is constrained by the finite motion range, which is given by
$\big| x_{l,n} - x_{l,n}^{\mathrm{prev}} \big| \le vT_1,\, \forall l,n,$ where $v$ denotes the maximum moving speed.
Therefore, for the $(l,n)$-th PA, the locally feasible index set is defined as
\begin{equation}
\mathcal{I}_{l,n}^{\mathrm{local}} 
=  \left\{ i \in \mathcal{I}^{\mathrm{global}} 
\,\middle|\,
|x^{(i)} - x_{l,n}^{\mathrm{prev}}| \le vT_1
\right\},
\end{equation}
and the corresponding locally feasible position set is given by
\begin{equation}
\mathcal{X}_{l,n}^{\mathrm{local}} 
= \{ x^{(i)} \mid i \in \mathcal{I}_{l,n}^{\mathrm{local}} \}.
\end{equation}
Each PA selects exactly one mounting point from its locally feasible set. To characterize the discrete position decision, we introduce binary variables $\phi_{l,n,i} \in \{0,1\}$, $\forall i \in \mathcal{I}^{\mathrm{global}}$, where $\phi_{l,n,i}=1$ indicates that the $(l,n)$-th PA is assigned to the mounting point indexed by $i$. 
Note that $\phi_{l,n,i}=0$ for all $i \notin \mathcal{I}_{l,n}^{\mathrm{local}}$. 
Accordingly, the position of the $(l,n)$-th PA is expressed as
\begin{equation}
x_{l,n} = \sum_{i \in \mathcal{I}_{l,n}^{\mathrm{local}}} \phi_{l,n,i} x^{(i)}.
\end{equation}
Moreover, adjacent PAs maintain a minimum separation $\Delta_{\min}$ to avoid mutual coupling among antennas, i.e., $x_{l+1,n}-x_{l,n}\geq \Delta_{\min},\, \forall l \in \mathcal{L}' =\{1,\ldots,L-1\}, n\in\mathcal N$.
Finally, each PA occupies exactly one feasible mounting point, while each mounting point can host at most one PA on each waveguide, yielding $\sum_{i \in \mathcal{I}_{l,n}^{\mathrm{local}}} \phi_{l,n,i} = 1, \, \forall l,n,$ and $\sum_{l=1}^{L} \phi_{l,n,i} \le 1, \, \forall n \in\mathcal N,\, i \in \mathcal{I}^{\mathrm{global}}$. Accordingly, the spatial position of the $(l,n)$-th PA is given by $\boldsymbol{\eta}_{l,n} = [x_{l,n},\,y_n^{\mathrm{W}},\,h]^T$.

\vspace{-1.5em}
\subsection{Signal Model}

In this subsection, we present the signal model for the general MWMU PASS.

\subsubsection{In-Waveguide Propagation}
\label{subsec:in_wg_propagation}
Let $\widetilde{c}_k$ denote the information symbols intended for user $k$ satisfying $\mathbb{E}[\widetilde{c}_k^H \widetilde{c}_k]=1$ and $\bm{\widetilde{c}} = [\widetilde{c}_1,\dots,\widetilde{c}_K]^T \in \mathbb{C}^{K\times1} $.
We denote the transmit beamforming matrix as $\mathbf W=\left[\mathbf w_{1},\dots, \mathbf w_{K}\right]\in\mathbb{C}^{N\times K}$, where $\mathbf w_{k} = [ w_{1,k},\dots, w_{N,k}]^T \in\mathbb{C}^{N\times1} $ denotes the transmit beamforming vector for user $k$. Therefore, the baseband signal injected into waveguide $n$ is given by
\begin{equation}
    c_n = \sum_{k\in\mathcal{K}} w_{n,k}\widetilde{c}_k, \forall n\in\mathcal{N}.
\end{equation}

Upon injection, the signal propagates along the direction of the dielectric waveguide, simultaneously accumulating phase and gradually transferring power to PAs. For the PA indexed by $l$ on waveguide $n$, the in-waveguide propagation distance from the feed point $\boldsymbol{\eta}_n^{\mathrm{W}}=[x_n^{\mathrm{W}},y_n^{\mathrm{W}},h]^T$
to the mounting position $\boldsymbol{\eta}_{l,n}=[x_{l,n},y_n^{\mathrm{W}},h]^T$ is $d^{\mathrm{wg}}_{l,n} = \big|x_{l,n} - x_n^{\mathrm{W}}\big|$.

We explicitly account for in-waveguide attenuation and phase rotation \cite{xuPinchingantennaSystemsInwaveguide2026}. Due to material properties and fabrication imperfections in real hardware, the guided mode experiences exponential attenuation and phase rotation over the distance. Consequently, the in-waveguide propagation factor from the feed point of waveguide $n$ to the $l$-th PA is defined as
\begin{equation}\label{eq:wg_gain}
    \widetilde{g}_{l,n}= \exp\!\left(-\alpha_{\mathrm{w}} d^{\mathrm{wg}}_{l,n}\right) \exp\!\left(-j\frac{2\pi}{\lambda_{\mathrm{w}}}d^{\mathrm{wg}}_{l,n}\right),
\end{equation}
where $\alpha_{\mathrm{w}}$ is the in-waveguide attenuation coefficient, which depends on the waveguide material, geometry, and operating frequency \cite{jiUltralowlossOnchipResonators2017,bautersPlanarWaveguidesLess2011}. The waveguide wavelength is given by $\lambda_{\mathrm{w}} = \lambda_f / n_{\mathrm{eff}}$, where $\lambda_f = \frac{c}{f_c}$ is the carrier wavelength, $c$ is the speed of light, $f_c$ is the carrier frequency, and $n_{\mathrm{eff}}$ is the effective refractive index of the dielectric waveguide.

Combining \eqref{eq:beta_discrete} and \eqref{eq:wg_gain}, the in-waveguide channel vector $\mathbf{g}_n \in \mathbb{C}^{L\times 1}$ is given by
\begin{equation}\label{eq:gln_def}
   \mathbf{g}_n =  [ \beta_{1,n} \widetilde{g}_{1,n}, \dots, \beta_{L,n} \widetilde{g}_{L,n}]^T.
\end{equation}
As each waveguide is connected to a subset of PAs, the overall in-waveguide channel matrix $\mathbf{G} \in\mathbb{C}^{M\times N}$ is given by $\mathbf{G}= \mathrm{blkdiag} \left(\mathbf{g}_{1},\mathbf{g}_{2}, \dots, \mathbf{g}_{N}\right)$.

\subsubsection{Free-Space Propagation}
We focus on LoS-dominant propagation between the PAs and the users.  The channel coefficient from the $l$-th PA on waveguide $n$ to user $k$ is modeled by the spherical-wave free-space expression
\begin{equation}
 h_{l,n,k}^{H}  =  \frac{\lambda_f}{4\pi \big\|\bm{\eta}_k^{\mathrm{U}}-\bm{\eta}_{l,n}\big\|}
   \exp\!\left(-j\frac{2\pi}{\lambda_{f}}\big\|\bm{\eta}_k^{\mathrm{U}}-\bm{\eta}_{l,n}\big\|\right).
\end{equation}

Let $\mathbf{h}_{n,k}^{H}=[h_{1,n,k}^{H},\dots,h_{L,n,k}^{H}] \in \mathbb{C}^{1\times L}$ denote the channel vector from the PAs on waveguide $n$ to user $k$, then the free-space channel vector is $\mathbf{h}_k^{H} = [\mathbf{h}_{1,k}^{H},\dots,\mathbf{h}_{N,k}^{H}]\in\mathbb{C}^{1\times M}$.
Therefore, the overall received signal at user $k$ is given by
\begin{align}
y_k =\underbrace{\mathbf{h}_k^H\mathbf{G} \mathbf w_k \widetilde{c}_k }_{\text{Desired signal}} + \underbrace{\underset{k' \in \mathcal{K}\setminus \{k\}}{\sum} \mathbf{h}_k^H\mathbf{G}\mathbf w_{k'} \widetilde{c}_{k'}}_{\text{Multi-user interference}} + n_k,
\end{align}
where $n_k\sim\mathcal{CN}(0,\sigma^2)$ is additive white Gaussian noise. Correspondingly, the signal-to-interference-plus-noise ratio (SINR) of user $k$ can be expressed as
\begin{equation}
 \mathrm{SINR}_k 
  = \frac{\big|\mathbf{h}_k^{H}\mathbf{G}\mathbf w_k\big|^2}   {\sum_{k'\neq k}\big|\mathbf{h}_k^{H}\mathbf{G}  \mathbf w_{k'}\big|^2+\sigma^2}.
\end{equation}

\vspace{-1.5em}
\subsection{Power Consumption for PASS}
The total average power consumption of the considered PASS system over one frame of duration $T_1+T_2$ consists of two components: the average transmit power and the average PA movement power. The mechanical movement of the PAs along the waveguides is driven by electric motors (e.g., stepper motors), which consume a constant power denoted by $P_{\mathrm{m}}$ and support the movement speed $v$. Since the movement distance of the $(l,n)$-th PA within one frame is $ d_{l,n}^{\mathrm{mov}} = \big| x_{l,n} - x_{l,n}^{\text{prev}}\big|$, the motion energy consumption is given by \cite{wuGloballyOptimalMovable2025}:
\begin{equation}\label{eq:Emn_def}
    E_{l,n}  = P_{\mathrm{m}}\, t_{l,n}^{\mathrm{mov}}= \frac{P_{\mathrm{m}}}{v} \big| x_{l,n} - x_{l,n}^{\text{prev}}\big|.
\end{equation}

The total motion energy consumption of all PAs within one frame is given by $E_{\mathrm{m}}  = \sum_{n\in\mathcal{N}}\sum_{l\in\mathcal{L}} E_{l,n}$.
Correspondingly, the average motion power over the whole frame is expressed as $P_{\mathrm{mot}} = \frac{1}{T}     \sum_{n\in\mathcal{N}}\sum_{l\in\mathcal{L}} E_{l,n}$.
Since downlink transmission only occurs in the second phase $T_2$,
the average transmit power is given by $ P_{\mathrm{t}}   = \frac{T_2}{T}  \sum_{k\in\mathcal{K}}\|\mathbf w_k\|_2^2$.

Therefore, the total average power consumption can be expressed as
\begin{equation}\label{eq:Ptotal_def}
    P_{\mathrm{total}}
    = \frac{T_2}{T}
      \sum_{k\in\mathcal{K}}\|\mathbf w_k\|_2^2
      + \frac{1}{T}
        \sum_{n\in\mathcal{N}}\sum_{l\in\mathcal{L}} E_{l,n}.
\end{equation}

\vspace{-1.75em}
\subsection{Problem Formulation}
\vspace{-0.25em}
In this paper, we aim to minimize the average power consumption while guaranteeing the QoS requirements by jointly optimizing the discrete PA position variables $\bm \Phi$, the coupling spacing selection variables $\mathbf V$, and the transmit beamforming matrix $\mathbf W$. The problem is formulated as
\begin{align}
%\mathcal P_{\mathrm{M}}:\ 
\mathcal P_0:\ 
\underset{\mathbf \Phi,\mathbf V,\mathbf W}{\min}\quad 
& P_{\rm total} \label{P0_obj}\\
{\rm s.t.}\quad 
& {\rm SINR}_k\ge \Gamma_k,\, \forall k\in\mathcal K, \tag{\ref{P0_obj}a} \label{SINR}\\
& \sum_{i \in \mathcal{I}_{l,n}^{\mathrm{local}}} \phi_{l,n,i} = 1, \, \forall l,n,\tag{\ref{P0_obj}b} \label{cons:P0_one_point}\\
& \sum_{l=1}^{L} \phi_{l,n,i} \le 1, \, \forall n,\, i \in \mathcal{I}^{\mathrm{global}}, \tag{\ref{P0_obj}c} \label{cons:P0_one_pa_per_point}\\
& \sum_{q=1}^{Q}v_{l,n,q}= 1,\, \forall l\in\mathcal L,\,n\in\mathcal N, \tag{\ref{P0_obj}d} \label{cons:P0_spacing_sel}\\
& \Big|x_{l,n}-x^{\rm prev}_{l,n}\Big|\le vT_1,\ \forall l\in\mathcal L,\,n\in\mathcal N, \tag{\ref{P0_obj}e} \label{cons:P0_motion}\\
&x_{l+1,n} - x_{l,n} \ge \Delta_{\min},  \forall l \in\mathcal L',\, \forall n \in \mathcal{N}, \tag{\ref{P0_obj}f} \label{cons:P0_min_spacing}\\
& \phi_{l,n,i}\in\{0,1\}, v_{l,n,q}\in\{0,1\},\ \forall l, n, i,q, \, \tag{\ref{P0_obj}g} \label{cons:P0_phi_bin}
\end{align}
where $\bm\Phi\triangleq\{\phi_{l,n,i}\}$ and $\mathbf V\triangleq\{v_{l,n,q}\}$ denote the sets of binary PA position and coupling spacing selection variables, respectively. The parameter $\Gamma_k$ denotes the SINR requirement of user $k$. Constraint \eqref{SINR} guarantees the QoS requirement of each user. Constraint \eqref{cons:P0_one_point} and \eqref{cons:P0_one_pa_per_point} ensure valid PA position selection, where each PA occupies exactly one mounting point and each mounting point hosts at most one PA. Constraint \eqref{cons:P0_spacing_sel} guarantees that each PA selects one coupling spacing level for radiation control. Constraint \eqref{cons:P0_motion} limits the PA movement distance within one time frame according to the mechanical motion capability. Constraint \eqref{cons:P0_min_spacing} ensures that the minimum distance between adjacent PAs per waveguide is no smaller than the predefined threshold $\Delta_{\min}$. Constraint \eqref{cons:P0_phi_bin} enforces the binary nature of the antenna position selection and spacing selection variables.

\vspace{-1.0em}
\section{BnB-Based Algorithm for the MWSU Scenario}
\label{sec:BnB_MWSU}

In this section, we propose a BnB-based algorithm for the MWSU scenario. In this case, the BS serves a single user through multiple waveguides. To solve the NP-hard problem, we adopt the BnB algorithm to optimize the discrete PA movement and coupling spacing variables. For any given PA positions and coupling spacing selections, the transmit beamforming vector admits a closed-form MRT solution.

\vspace{-1.25em}
\subsection{Problem Formulation}

For the MWSU scenario, problem $\mathcal P_0$ in \eqref{P0_obj} is reduced to
\begin{align}
\mathcal{P}_{\mathrm{S}}:\ 
\underset{\bm{\Phi},\bm{V},\mathbf w}{\min}\quad 
& P_{\mathrm{total}}^{\mathrm{S}}
= \frac{T_2}{T} \|\mathbf w\|_2^2
+  \frac{1}{T} \sum_{n\in\mathcal{N}}\sum_{l\in\mathcal{L}} E_{l,n}
\label{P_MWSU_obj}\\
\mathrm{s.t.}\quad
& \frac{\left|\mathbf h^H \mathbf G \mathbf w\right|^2}{\sigma^2} \ge \Gamma,
\tag{\ref{P_MWSU_obj}a}\label{cons:MWSU_SNR}\\
& \eqref{cons:P0_one_point}, \,\eqref{cons:P0_one_pa_per_point}, \,\eqref{cons:P0_spacing_sel},\,\eqref{cons:P0_motion}, \,\eqref{cons:P0_min_spacing}, \,\eqref{cons:P0_phi_bin}.\notag
\end{align}

Let $\mathbf x=[x_{1,1},\ldots,x_{L,N}]^T \in\mathbb R^{LN\times1}$, $\mathbf s=[s_{1,1},\ldots,s_{L,N}]^T \in\mathbb R^{LN\times1}$ denote the stacked PA position vector and coupling spacing vector, respectively.
For given $\mathbf x$ and $\mathbf s$, define the equivalent channel vector
\begin{equation}
\mathbf c(\mathbf x,\mathbf s)
=
[c_1(\mathbf x,\mathbf s),\ldots,c_N(\mathbf x,\mathbf s)]^T
\in\mathbb C^{N\times1},
\end{equation}
where $\mathbf c^H(\mathbf x,\mathbf s) = \mathbf h^H\mathbf G$.
Since $\mathbf G=\mathrm{blkdiag}(\mathbf g_1,\ldots,\mathbf g_N)$ and $\mathbf h^H=[\mathbf h_1^H,\ldots,\mathbf h_N^H]$, the $n$-th waveguide effective channel is given by
\begin{equation}
c_n(\mathbf x,\mathbf s)
=\mathbf h_n^H\mathbf g_n=
\sum_{l=1}^{L}
\beta_{l,n}(\mathbf s_n)
\xi_{l,n}(x_{l,n}),
\label{eq:cn_MWSU}
\end{equation}
where $\mathbf s_n=[s_{1,n},\ldots,s_{L,n}]^T$ denotes the coupling spacing vector on waveguide $n$, $\beta_{l,n}(\mathbf s_n)$ is the cascaded radiation coefficient of PA $(l,n)$, and $\xi_{l,n}(x_{l,n}) \triangleq h_{l,n}^{H}(x_{l,n}) \widetilde g_{l,n}(x_{l,n})$ is the position-dependent cascaded propagation coefficient of PA $(l,n)$. Therefore, we have
\begin{equation}
\|\mathbf c(\mathbf x,\mathbf s)\|_2^2
=
\sum_{n=1}^{N}
\left|
\sum_{l=1}^{L}
\beta_{l,n}(\mathbf s_n)
\xi_{l,n}(x_{l,n})
\right|^2.
\label{eq:channel_norm_MWSU}
\end{equation}

Based on the above auxiliary variables, the SNR constraint in \eqref{cons:MWSU_SNR} becomes
$\frac{\left|\mathbf c^H(\mathbf x,\mathbf s)\mathbf w\right|^2}{\sigma^2}\ge \Gamma$.
For fixed $\mathbf x$ and $\mathbf s$, the transmit beamforming subproblem is
\begin{equation}
\min_{\mathbf w}\ \|\mathbf w\|_2^2
\quad
\mathrm{s.t.}\quad
\left|\mathbf c^H(\mathbf x,\mathbf s)\mathbf w\right|^2
\ge \Gamma\sigma^2.
\label{eq:beamforming_subproblem_MWSU}
\end{equation}
The optimal solution follows MRT as
\begin{equation}
\mathbf w^*
=\sqrt{\Gamma\sigma^2}
\frac{\mathbf c(\mathbf x,\mathbf s)}
{\|\mathbf c(\mathbf x,\mathbf s)\|_2^2},
\label{eq:MRT_solution_MWSU}
\end{equation}
 and the corresponding minimum transmit power is
$\|\mathbf w^*\|_2^2= \frac{\Gamma\sigma^2}{\|\mathbf c(\mathbf x,\mathbf s)\|_2^2}$.

Accordingly, problem $\mathcal P_{\mathrm S}$ is equivalently transformed into
\begin{align}
\mathcal P_{\mathrm S}':\quad
\min_{\mathbf x,\mathbf s}\quad
&
P_{\mathrm{total}}^{\mathrm S}(\mathbf x,\mathbf s)
=
\frac{T_2}{T}
\frac{\Gamma\sigma^2}
{\|\mathbf c(\mathbf x,\mathbf s)\|_2^2}
+
\frac{1}{T}
\sum_{n\in\mathcal N}\sum_{l\in\mathcal L} E_{l,n}
\label{eq:MWSU_Ptotal_BnB}
\\
\mathrm{s.t.}\quad
&
x_{l,n}\in\mathcal{X}_{l,n}^{\mathrm{local}},
\,\forall l\in\mathcal L,\, n\in\mathcal N,\tag{\ref{eq:MWSU_Ptotal_BnB}a}
\label{cons:MWSU_x_set}
\\
&
s_{l,n}\in\mathcal S,
\,\forall l\in\mathcal L,\, n\in\mathcal N,\tag{\ref{eq:MWSU_Ptotal_BnB}b}
\label{cons:MWSU_s_set}
\\
&\eqref{cons:P0_motion}, \,\eqref{cons:P0_min_spacing}. \notag
\end{align}
Therefore, the continuous beamforming variable can be eliminated in closed form, and the resulting problem only involves the discrete PA movement and coupling spacing variables.

\vspace{-1.25em}
\subsection{Branch-and-Bound Algorithm}

In the BnB procedure, each node $\mathcal B$ is represented by the remaining candidate position and coupling spacing domains
\begin{equation}
\mathcal B =
\left\{
\mathcal X_{l,n}(\mathcal B),
\mathcal S_{l,n}(\mathcal B),
\forall l,n
\right\},
\end{equation}
where $\mathcal X_{l,n}(\mathcal B)$ and $\mathcal S_{l,n}(\mathcal B)$ denote the candidate position set and coupling spacing set of PA $(l,n)$ at node $\mathcal B$, respectively.
\subsubsection{Lower Bound}

By the triangle inequality, we have
\begin{equation}
|c_n(\mathbf x,\mathbf s)|
\le
\sum_{l=1}^{L}
|\beta_{l,n}(\mathbf s_n)|
|\xi_{l,n}(x_{l,n})|.
\end{equation}

For node $\mathcal B$, with $x\in\mathcal X_{l,n}(\mathcal B)$, define
\begin{equation}
r_{l,n}(\mathcal B) =
\max_{x}
|\xi_{l,n}(x)|,\,
d_{l,n}^{\min}(\mathcal B)=
\min_{x}\left|x-x_{l,n}^{\mathrm{prev}}\right|.
\end{equation}

For waveguide $n$, an optimistic upper bound on the magnitude of its effective channel is obtained by enumerating all remaining coupling spacing combinations:
\begin{equation}
C_n^{\mathrm{UB}}(\mathcal B) =
\max_{\substack{s_{l,n}\in\mathcal S_{l,n}(\mathcal B),\\ l=1,\ldots,L}}
\sum_{l=1}^{L}
|\beta_{l,n}(\mathbf s_n)|
\, r_{l,n}(\mathcal B).
\label{eq:Cn_UB}
\end{equation}

Then, an upper bound on the MRT channel gain is given by
\begin{equation}
\|\mathbf c(\mathbf x,\mathbf s)\|_2^2
\le
C^{\mathrm{UB}}(\mathcal B)
\triangleq
\sum_{n=1}^{N}
\left(C_n^{\mathrm{UB}}(\mathcal B)\right)^2.
\label{eq:C_UB}
\end{equation}

Accordingly, a valid lower bound of the objective value over node $\mathcal B$ is
\begin{equation}
f_{\mathrm{LB}}(\mathcal B)
=
\frac{T_2}{T}
\frac{\Gamma\sigma^2}
{C^{\mathrm{UB}}(\mathcal B)}
+
\frac{P_{\mathrm m}}{vT}
\sum_{n=1}^{N}\sum_{l=1}^{L}
d_{l,n}^{\min}(\mathcal B).
\label{eq:MWSU_LB}
\end{equation}

\vspace{-0.25em}
\subsubsection{Upper Bound}

The upper bound is obtained from a feasible solution  within node $\mathcal B$. Specifically, the current incumbent solution is first adjusted to satisfy the remaining position and coupling spacing domains of $\mathcal B$. Starting from this projected solution, coordinate descent is performed by alternately updating the coupling spacing and the position of each PA. For each update, the candidate yielding the smallest objective value in \eqref{eq:MWSU_Ptotal_BnB} is selected. In addition, several random feasible starts are refined in the same manner. Let $(\widehat{\mathbf x},\widehat{\mathbf s})$ denote the best feasible solution obtained in node $\mathcal B$. Then, $f_{\mathrm{UB}}(\mathcal B)= P_{\mathrm{total}}^{\mathrm S}(\widehat{\mathbf x},\widehat{\mathbf s})$
is a valid upper bound. The corresponding transmit beamforming vector is recovered from \eqref{eq:MRT_solution_MWSU}.

The proposed BnB algorithm follows a best-bound-first search rule. Specifically, at each iteration, the active node with the smallest local lower bound is selected. Let $\mathrm{GUB}$ denote the current global upper bound, i.e., the best feasible objective value obtained so far. If the selected node $\mathcal B$ satisfies $f_{\mathrm{LB}}(\mathcal B)\ge \mathrm{GUB}$,
then $\mathcal B$ is discarded, since all feasible configurations contained in this node have objective values no smaller than the current incumbent. Otherwise, a feasible upper-bound solution is constructed, and $\mathrm{GUB}$ is updated if a smaller objective value is obtained. The selected node is then further partitioned. In the branching step, the coupling spacing domains are first considered, since the spacing variables determine the cascaded radiation coefficients. If all spacing domains have been fixed, the algorithm branches over the position domain with the largest number of remaining candidates.
The algorithm terminates when $\frac{\mathrm{GUB}-\mathrm{GLB}} {\max\{1,|\mathrm{GUB}|\}} \le \epsilon$, where $\mathrm{GLB}=\min_{\mathcal B\in\mathcal A}f_{\mathrm{LB}}(\mathcal B)$ denotes the smallest lower bound among all active nodes in the active-node set $\mathcal A$. The following theorem establishes the finite convergence and optimality guarantee of the proposed BnB algorithm. The overall procedure of the proposed BnB algorithm for solving problem $\mathcal P_{\mathrm S}$ is summarized in Algorithm~\ref{alg:BnB_MWSU}.

\begin{theorem}\label{Theorem:BnB}
 The proposed BnB algorithm converges to an $\epsilon$-optimal solution of problem $\mathcal P_{\mathrm S}'$ in finite iterations.
\end{theorem}
\begin{IEEEproof}
    See Appendix \ref{proof:Theorem}.
\end{IEEEproof}

\begin{algorithm}[t]
\caption{BnB-Based Algorithm for the MWSU Scenario}
\label{alg:BnB_MWSU}
\begin{algorithmic}[1]
\Require Feasible PA position sets and coupling spacing sets
\Ensure $\epsilon$-optimal solution $(\mathbf x^\star,\mathbf s^\star,\mathbf w^\star)$

\State Initialize the root node $\mathcal B_0$ and set the active node set as $\mathcal A=\{\mathcal B_0\}$
\State Compute $f_{\mathrm{LB}}(\mathcal B_0)$ and a feasible upper-bound solution of $\mathcal B_0$
\State Initialize $\mathrm{GUB}$ as the best feasible objective value and set $\mathrm{GLB}=f_{\mathrm{LB}}(\mathcal B_0)$

\While{$\mathcal A\neq\emptyset$ and 
$\frac{\mathrm{GUB}-\mathrm{GLB}}{\max\{1,|\mathrm{GUB}|\}}>\epsilon$}
    \State Select the active node $\mathcal B$ with the smallest lower bound and remove it from $\mathcal A$
    \If{$f_{\mathrm{LB}}(\mathcal B)\ge \mathrm{GUB}$}
        \State Discard node $\mathcal B$
    \Else
        \State Compute a feasible upper-bound solution within $\mathcal B$
        \State Update the incumbent solution and $\mathrm{GUB}$ if a smaller objective value is obtained
        \State Branch node $\mathcal B$ over the remaining coupling spacing or position domains
        \State Compute the lower bounds of the generated child nodes and insert the non-pruned child nodes into $\mathcal A$
    \EndIf
    \State Update $\mathrm{GLB}=\min_{\mathcal B'\in\mathcal A} f_{\mathrm{LB}}(\mathcal B')$
\EndWhile

\State Recover the MRT beamformer $\mathbf w^\star$ according to \eqref{eq:MRT_solution_MWSU}
\State Output $(\mathbf x^\star,\mathbf s^\star,\mathbf w^\star)$
\end{algorithmic}
\end{algorithm}

The computational complexity of the proposed BnB algorithm mainly depends on the number of nodes explored during the search. Let $I_{\mathrm m}=\max_{l,n}|\mathcal I_{l,n}^{\mathrm{local}}|$ denote the maximum number of feasible positions for a single PA, and let $M=LN$ denote the total number of PAs. Since each PA has at most $I_{\mathrm m}$ feasible positions and $Q$ coupling spacing levels, the exhaustive search space contains at most $(I_{\mathrm m}Q)^M$ feasible configurations.
For each processed node, the lower-bound evaluation requires computing the propagation-related bounds over the feasible position domains and enumerating the feasible coupling spacing combinations on each waveguide, which leads to the complexity
$\mathcal O(NL I_{\mathrm m}+NQ^L)$.
The feasible upper-bound solution is obtained by coordinate-descent refinement. With $T_{\mathrm{cd}}$ denoting the number of coordinate-descent iterations, the corresponding complexity is approximately $\mathcal O\left(T_{\mathrm{cd}}M(I_{\mathrm m}+Q)\right)$.
Therefore, the total computational complexity can be expressed as $\mathcal O\left(N_{\mathrm{node}}\left(NL I_{\mathrm m}+NQ^L+T_{\mathrm{cd}}M(I_{\mathrm m}+Q)\right)\right)$.
In the worst case, the BnB search may degenerate into exhaustive enumeration over all feasible discrete configurations, whose size is on the order of
$\mathcal O\left((I_{\mathrm m}Q)^M\right)$.
Nevertheless, in practice, $N_{\mathrm{node}}$ is usually much smaller than $(I_{\mathrm m}Q)^M$ due to the proposed lower-bound pruning.

\section{GA-PSO-based Algorithm for the MWMU scenario}
\label{sec:GA-PSO_SOCP}

Problem $\mathcal P_0$ is a MINLP problem due to the binary PA position variables $\boldsymbol{\Phi}$, the binary coupling spacing selection variables $\mathbf V$, and the continuous transmit beamforming matrix $\mathbf W$. Moreover, these variables are strongly coupled in the equivalent PASS channel and the SINR constraints, which makes the direct joint optimization computationally prohibitive. To obtain high-quality feasible solutions with scalable complexity, we propose a GA-assisted particle swarm optimization framework with SOCP-based beamforming reconstruction, termed the GA-PSO algorithm. The proposed algorithm adopts a two-layer structure. In the outer layer, a block-structured GA-PSO algorithm searches over the discrete PA position and radiation-control variables through continuous latent representations. In the inner layer, for each decoded PASS structural configuration, the transmit beamforming matrix is optimally reconstructed by solving an SOCP problem.

\vspace{-1.05em}
\subsection{SOCP-Based Transmit Beamforming Optimization}

For given PA position and radiation-control configurations $(\boldsymbol{\Phi},\mathbf V)$, the original problem $\mathcal P_0$ reduces to the following transmit beamforming optimization problem:
\begin{align}
\min_{\mathbf W}\quad
& \sum_{k\in\mathcal K}\|\mathbf w_k\|_2^2
\label{eq:beamforming_sub}\\
\mathrm{s.t.}\quad
& \frac{|\mathbf h_k^H\mathbf G\mathbf w_k|^2}
{\sum_{j\neq k}|\mathbf h_k^H\mathbf G\mathbf w_j|^2+\sigma^2}
\geq \Gamma_k, \forall k.
\tag{\ref{eq:beamforming_sub}a}
\end{align}

Let $\widetilde{\mathbf h}_k^H \triangleq \mathbf h_k^H\mathbf G$ denote the equivalent cascaded channel for user $k$. By exploiting the phase-rotation invariance of the SINR constraints, the desired signal term $\widetilde{\mathbf h}_k^H\mathbf w_k$ can be constrained to be real-valued without loss of optimality. Therefore, problem \eqref{eq:beamforming_sub} can be equivalently reformulated as the following SOCP:
\begin{align}
\underset{\{\mathbf w_k\}_{k=1}^K}{\min}\quad
& \sum_{k\in\mathcal K}\|\mathbf w_k\|_2^2
\label{eq:SOCP_sub_obj}\\
\mathrm{s.t.}\quad
& \Re\{\widetilde{\mathbf h}_k^H\mathbf w_k\}
\geq
\sqrt{\Gamma_k}
\left\|
\left[
\{\widetilde{\mathbf h}_k^H\mathbf w_j\}_{j\neq k},
\sigma
\right]
\right\|_2,\, \forall k,
\tag{\ref{eq:SOCP_sub_obj}a}
\label{eq:SOCP_sub_soc}\\
& \Im\{\widetilde{\mathbf h}_k^H\mathbf w_k\}=0,
\,  \forall k.
\tag{\ref{eq:SOCP_sub_obj}b}
\label{eq:SOCP_sub_imag}
\end{align}
Problem \eqref{eq:SOCP_sub_obj} is a standard SOCP and can be efficiently solved by convex optimization solvers such as CVX. Hence, the proposed algorithm only performs heuristic search over the discrete PASS structural variables, while the continuous transmit beamforming is optimally reconstructed for each candidate configuration.

\subsection{Particle Encoding and Discrete Decoding}

The proposed GA-PSO algorithm jointly optimizes the PA position variables $\boldsymbol{\Phi}$ and the radiation-control variables $\mathbf V$. Each particle $r$ at iteration $t$ is represented by two continuous latent blocks:
\begin{equation}
    \mathbf z_r^{(t)}
    =
    \left[
    \mathbf z_{r,X}^{(t)},
    \mathbf z_{r,V}^{(t)}
    \right],
\end{equation}
where $\mathbf z_{r,X}^{(t)}\in\mathbb R^{L\times N}$ denotes the PA-position proxy block, and $\mathbf z_{r,V}^{(t)}\in\mathbb R^{L\times N\times Q}$ denotes the radiation-control score block.

For the position block, each continuous proxy $z_{r,X,l,n}^{(t)}$ is first clipped into the local feasible movement interval of the $(l,n)$-th PA and then projected onto the nearest feasible mounting point in $\mathcal X_{l,n}^{\mathrm{local}}$. A sorting-and-repair procedure is further applied on each waveguide to satisfy the minimum spacing constraint in \eqref{cons:P0_min_spacing}. The decoded discrete position index is denoted by $i_{l,n}$, and the corresponding binary position variable is constructed as
\begin{equation}
    \phi_{l,n,i} =
    \begin{cases}
    1, & i=i_{l,n},\\
    0, & \text{otherwise}.
    \end{cases}
\end{equation}

For the radiation-control block, the selected coupling spacing level is determined by the maximum-score criterion:
\begin{equation}
    q_{l,n}^{(t)}   =
    \arg\max_{q\in\mathcal Q}
    z_{r,V,l,n,q}^{(t)}.
\end{equation}
Then, the corresponding binary radiation-control variable is given by
\begin{equation}
    v_{l,n,q}
    =
    \begin{cases}
    1, & q=q_{l,n}^{(t)},\\
    0, & \text{otherwise}.
    \end{cases}
\end{equation}
Based on the decoded spacing level, the coupling spacing, coupling coefficient, and cascaded radiation ratio are obtained according to \eqref{eq:kappa_discrete} and \eqref{eq:beta_discrete}.

\subsection{Block-Structured PSO Update}

Let $\mathbf u_{r,X}^{(t)}$ and $\mathbf u_{r,V}^{(t)}$ denote the velocities of the position and radiation-control blocks, respectively. Let $\mathbf p_{r,X}$ and $\mathbf p_{r,V}$ denote the personal-best latent blocks of particle $r$, and let $\mathbf g_X$ and $\mathbf g_V$ denote the global-best latent blocks. The PSO update for the position block is given by
\begin{align}
\mathbf u_{r,X}^{(t+1)}
=&
\omega_X\mathbf u_{r,X}^{(t)}
+
c_{1,X}\mathbf R_{1,X}^{(t)}
\odot
\left(
\mathbf p_{r,X}-\mathbf z_{r,X}^{(t)}
\right)
\notag\\
&+
c_{2,X}\mathbf R_{2,X}^{(t)}
\odot
\left(
\mathbf g_X-\mathbf z_{r,X}^{(t)}
\right),
\end{align}
where $\omega_X$ is the inertia weight, $c_{1,X}$ and $c_{2,X}$ are the cognitive and social coefficients, $\mathbf R_{1,X}^{(t)}$ and $\mathbf R_{2,X}^{(t)}$ are random matrices with entries uniformly distributed in $[0,1]$, and $\odot$ denotes element-wise multiplication. The position latent block is then updated as
\begin{equation}
    \mathbf z_{r,X}^{(t+1)}
    = \mathbf z_{r,X}^{(t)} + \mathbf u_{r,X}^{(t+1)}.
\end{equation}

Similarly, the radiation-control block is updated as
\begin{align}
\mathbf u_{r,V}^{(t+1)}
=&
\omega_V\mathbf u_{r,V}^{(t)}+c_{1,V}\mathbf R_{1,V}^{(t)}
\odot
\left(
\mathbf p_{r,V}-\mathbf z_{r,V}^{(t)}
\right)
\notag\\
&+
c_{2,V}\mathbf R_{2,V}^{(t)}
\odot
\left(
\mathbf g_V-\mathbf z_{r,V}^{(t)}
\right),
\end{align}
and
\begin{equation}
    \mathbf z_{r,V}^{(t+1)}
    =
    \mathbf z_{r,V}^{(t)}
    +
    \mathbf u_{r,V}^{(t+1)}.
\end{equation}

After the block-wise PSO update, each particle is decoded into $(\boldsymbol{\Phi},\mathbf V)$. Then, the cascaded channel matrix $\mathbf G$ is constructed, and the SOCP problem in \eqref{eq:SOCP_sub_obj} is solved to reconstruct the optimal transmit beamforming matrix $\mathbf W$. The fitness value of particle $r$ is defined as the resulting total average power consumption:
\begin{equation}
    F_r
    =
    \frac{T_2}{T}
    \sum_{k\in\mathcal K}\|\mathbf w_k^\star\|_2^2
    +
    \frac{P_{\mathrm m}}{vT}
    \sum_{n\in\mathcal N}\sum_{l\in\mathcal L}
    |x_{l,n}-x_{l,n}^{\mathrm{prev}}|.
\end{equation}
If the SOCP problem is infeasible or the decoded particle violates the discrete feasibility constraints, a sufficiently large penalty value is assigned to the particle. Based on the evaluated fitness values, the personal-best and global-best particles are updated according to
\begin{equation}
    \mathbf p_r
    \leftarrow
    \begin{cases}
    \mathbf z_r^{(t)}, & F_r^{(t)} < F_{p,r},\\
    \mathbf p_r, & \text{otherwise},
    \end{cases}
\end{equation}
and $\mathbf g  = \arg\min_{\mathbf p_r} F_{p,r}$, where $F_{p,r}$ denotes the personal-best fitness of particle $r$.

\subsection{GA-Assisted Offspring Generation}

To avoid premature convergence, a lightweight GA module is invoked every $T_{\rm GA}$ iterations after a warm-up period. The GA module generates $N_{\rm GA}$ offspring from elite or personal-best particles through crossover and mutation, and replaces inferior particles when improved fitness values are achieved.
First, parent particles are selected by tournament selection. For each parent selection, a subset of particles is randomly sampled, and the particle with the smallest fitness in the subset is selected. Let $(\mathbf z_{X}^{(1)},\mathbf z_{V}^{(1)})$ and $(\mathbf z_{X}^{(2)},\mathbf z_{V}^{(2)})$ denote two selected parents.

For the position block, waveguide-level crossover is used to preserve the structural ordering of PAs on each waveguide. Specifically, for each waveguide $n$, the offspring inherits the whole position vector of waveguide $n$ from the second parent with probability $p_{c,X}$:
\begin{equation}
    \mathbf z_{X,:,n}^{\rm child}
    =
    \begin{cases}
    \mathbf z_{X,:,n}^{(2)}, & \text{with probability } p_{c,X},\\
    \mathbf z_{X,:,n}^{(1)}, & \text{otherwise}.
    \end{cases}
\end{equation}
For the radiation-control block, gene-level crossover is adopted. For each PA, the corresponding spacing-score vector is inherited from the second parent with probability $p_{c,V}$:
\begin{equation}
    \mathbf z_{V,l,n,:}^{\rm child}
    =
    \begin{cases}
    \mathbf z_{V,l,n,:}^{(2)}, & \text{with probability } p_{c,V},\\
    \mathbf z_{V,l,n,:}^{(1)}, & \text{otherwise}.
    \end{cases}
\end{equation}

After crossover, mutation is applied to further diversify the offspring. For the position block, Gaussian perturbation is added with probability $p_{m,X}$:
\begin{equation}
    z_{X,l,n}^{\rm child}
    \leftarrow
    z_{X,l,n}^{\rm child}
    +
    \sigma_X\epsilon,\,
    \epsilon\sim\mathcal N(0,1).
\end{equation}
The mutated position proxy is then clipped into the local feasible interval and repaired through the discrete position decoder.

For the radiation-control block, mutation is applied with probability $p_{m,V}$. Let $q_0=\arg\max_{q\in\mathcal Q} z_{V,l,n,q}^{\rm child}$ be the current preferred spacing level. A new spacing level $\bar q$ is randomly selected from $\mathcal Q\setminus\{q_0\}$, and the score vector is perturbed as
\begin{equation}
    z_{V,l,n,\bar q}^{\rm child}
    \leftarrow
    z_{V,l,n,\bar q}^{\rm child}
    +
    \Delta_V,
\end{equation}
where $\Delta_V>0$ is a mutation strength used to promote the alternative spacing level. The offspring is then decoded into $(\boldsymbol{\Phi},\mathbf V)$, repaired if necessary, and evaluated through SOCP-based beamforming reconstruction. If the offspring achieves a smaller fitness value than the worst particle in the current swarm, it replaces that particle.

\begin{algorithm}[t]
\caption{Proposed GA-PSO Algorithm}
\label{alg:GA-PSO_SOCP}
\begin{algorithmic}[1]
\State Initialize particle population $\{\mathbf z_r^{(0)}\}_{r=1}^{N_{\rm p}}$, velocities $\{\mathbf u_r^{(0)}\}_{r=1}^{N_{\rm p}}$, and PSO parameters.
\State Decode and repair the initial particles to generate feasible PASS configurations.
\State Evaluate each particle by solving the SOCP problem in \eqref{eq:SOCP_sub_obj}.
\State Initialize personal-best particles $\{\mathbf p_r\}_{r=1}^{N_{\rm p}}$ and the global-best particle $\mathbf g$.
\For{$t=1,\ldots,I_{\max}$}
    \For{$r=1,\ldots,N_{\rm p}$}
        \State Update $\mathbf u_{r,X}^{(t)}$, $\mathbf z_{r,X}^{(t)}$, $\mathbf u_{r,V}^{(t)}$, and $\mathbf z_{r,V}^{(t)}$.
        \State Decode $\mathbf z_r^{(t)}$ into $(\boldsymbol{\Phi},\mathbf V)$.
        \State Repair infeasible PA positions and spacing-level selections.
        \State Construct the cascaded PASS channel matrix $\mathbf G$.
        \State Solve the SOCP problem in \eqref{eq:SOCP_sub_obj} to obtain $\mathbf W^\star$.
        \State Evaluate the fitness value $F_r^{(t)}$ using the total average power consumption.
        \State Update the personal-best particle $\mathbf p_r$.
    \EndFor
    \State Update the global-best particle $\mathbf g$.
    \If{$t$ satisfies the GA activation condition}
        \State Select parent particles by tournament selection.
        \State Generate offspring through structured crossover and mutation.
        \State Decode, repair, and evaluate offspring by SOCP.
        \State Replace inferior particles using elitist replacement.
    \EndIf
\EndFor
\State Output the optimized solution $(\boldsymbol{\Phi}^\star,\mathbf V^\star,\mathbf W^\star)$.
\end{algorithmic}
\end{algorithm}

\begin{remark}
The GA module is introduced to alleviate the premature convergence of PSO in the discrete PASS search space. Since the continuous latent particles are decoded into discrete PA positions and quantized coupling spacing levels, the resulting fitness landscape is highly discontinuous and multi-modal, making standard PSO prone to diversity loss and local stagnation. The crossover and mutation operations periodically inject structurally diverse offspring by recombining promising configurations and perturbing selected position or spacing variables. With elitist replacement, only improved offspring are retained, which enhances global exploration while preserving high-quality solutions.
\end{remark}

The overall GA-PSO algorithm is summarized in Algorithm~\ref{alg:GA-PSO_SOCP}.
The elitist update of the global-best particle ensures that the recorded best objective value of Algorithm~\ref{alg:GA-PSO_SOCP} is non-increasing over the iterations. Since the total average power consumption is lower bounded, the recorded global-best objective sequence is convergent. Moreover, for each decoded PASS configuration, the transmit beamforming subproblem is optimally solved by SOCP, which guarantees that the corresponding fitness value is evaluated with the optimal beamforming solution under the given discrete structural configuration. Nevertheless, the GA-PSO search over the discrete PA-position and coupling-spacing variables is heuristic, and thus global optimality is not claimed for the MWMU scenario. The proposed algorithm instead provides a scalable approach for obtaining high-quality feasible solutions with SOCP-based optimal beamforming reconstruction.

Let $N_{\rm p}$ denote the swarm size, $I_{\max}$ denote the maximum number of iterations, and $C_{\rm SOCP}$ denote the complexity of solving one SOCP problem in \eqref{eq:SOCP_sub_obj}. In the standard PSO update, each particle requires one SOCP evaluation per iteration. Moreover, if the GA module is activated every $T_{\rm GA}$ iterations and generates $N_{\rm GA}$ offspring at each activation, the additional SOCP evaluations caused by the GA module should also be considered. Therefore, the overall computational complexity can be approximated as $\mathcal O  \left(   I_{\max}N_{\rm p}C_{\rm SOCP} +   \frac{I_{\max}}{T_{\rm GA}}N_{\rm GA}C_{\rm SOCP} \right)$.
%The decoding, clipping, sorting, and repair operations only introduce low-order polynomial complexity with respect to $N$, $L$, and $Q$, and are therefore dominated by the repeated SOCP evaluations.

\vspace{-1.05em}
\section{Simulation Results}

In this section, numerical simulations are conducted to verify the effectiveness of the proposed PASS design. Unless otherwise specified, the simulation parameters are set according to \cite{xuJointRadiationPower2025}. The carrier frequency, noise power, in-waveguide attenuation coefficient, and effective refractive index are set as $f_c=28$ GHz, $\sigma^2=-80$ dBm, $\alpha_{\rm w}=0.01$, and $n_{\rm eff}=1.4$, respectively. A resource-limited scenario with $N=K=3$ and $L=4$ PAs per waveguide is considered.
The deployment region is $20$ m $\times$ $10$ m, while the BS height is fixed at $5$ m.
The minimum SINR requirement is $\gamma_{\min}=24$ dB. 
Moreover, three waveguides are deployed parallel to the $x$-axis, with feed-point coordinates given by $x_{n}^{\rm W}=0$ and $y_{n}^{\rm W}=\{0,5,10\}$. 
The durations for antenna movement and signal transmission are set as $T_1=0.2$ s and $T_2=0.8$ s, respectively. 
In addition, the motion power and antenna movement speed are set to $P_{\rm m}=0.1$ W and $v=1$ m/s, respectively.
The coupling model parameters are adopted from \cite{xuPinchingantennaSystemsPASS2025}.
A rectangular waveguide of cladding index $n_{\mathrm{clad}}=1.0$ and $2b=10$ mm is considered. 
The fitted parameters are $\alpha = 0.24615~\mathrm{mm}^{-1}$ and $\Omega_{0}=0.3300~\mathrm{mm}^{-1}$.
The effective pinching antenna length is set as $D^{\mathrm{PA}}=5~\mathrm{mm}$.
Accordingly, for $Q=6$, the set of spacing levels is $\mathcal S=\{0.1999,\,2.3626,\,3.8610,\,5.6664,\,8.5600,\,39.4572\}$ mm.

For performance comparison, the proposed scheme and the following baseline schemes are considered.
\begin{itemize} 
    \item \textbf{Adjustable coupling, discrete motion PASS (AC-DM-PASS, proposed):} 
    The proposed scheme jointly optimizes the localized discrete PA positions, adjustable coupling strength, and transmit beamforming. In particular, coupling strength is adopted to tune the radiation power at each PA through discrete coupling spacing levels.
    \item \textbf{Equal power, discrete motion PASS (DM-PASS):}  This baseline scheme adopts equal radiation power allocation for all PAs \cite{xuPinchingantennaSystemsPASS2025}. Specifically, the radiation power is uniformly allocated among the PAs, i.e., $\beta_{l,n}=\sqrt{\frac{1}{L}}$. The localized discrete PA positions and transmit beamforming are optimized via the proposed GA-PSO algorithm. 
   
    \item \textbf{Adjustable coupling, discrete activation PASS (DA-PASS):}  This baseline scheme considers a discrete activation PASS architecture with adjustable coupling strength, where all PAs are fixed at pre-mounted discrete positions  \cite{wangAntennaActivationResource2026}. The activation state of each PA is controlled by its coupling spacing to the waveguide, where an inactive PA corresponds to a near-zero radiation state. The discrete coupling spacing levels and transmit beamforming are jointly optimized using the proposed GA-PSO algorithm.
    %This baseline scheme considers a discrete activation PASS architecture with adjustable coupling strength and fixed discrete antenna positions. Specifically, all PAs are located at pre-mounted positions, and the activation state of each PA is controlled by its coupling spacing to the waveguide, where an inactive PA corresponds to a near-zero radiation state.
    \item \textbf{Hybrid MIMO:} A hybrid beamforming architecture \cite{songFullyPartiallyconnectedHybrid2019} is considered, where the BS is equipped with $N$ RF chains and each RF chain is connected to $L$ antennas through phase shifters. The hybrid beamforming coefficients are obtained by employing the penalty-based method proposed in \cite{shiSpectralEfficiencyOptimization2018}. Moreover, the conventional MIMO antenna array is deployed at the BS with half-wavelength antenna spacing. 
    \item \textbf{MIMO:} This baseline scheme corresponds to a conventional fully-digital MIMO system, where each RF chain is connected to a single antenna. Specifically, the BS is equipped with the same number of RF chains, i.e., $N$, as the considered PASS framework. 
\end{itemize}

\vspace{-1.05em}
\subsection{Single-User Scenario}

\begin{figure}[t]
  \centering
  \includegraphics[width=0.45\textwidth]{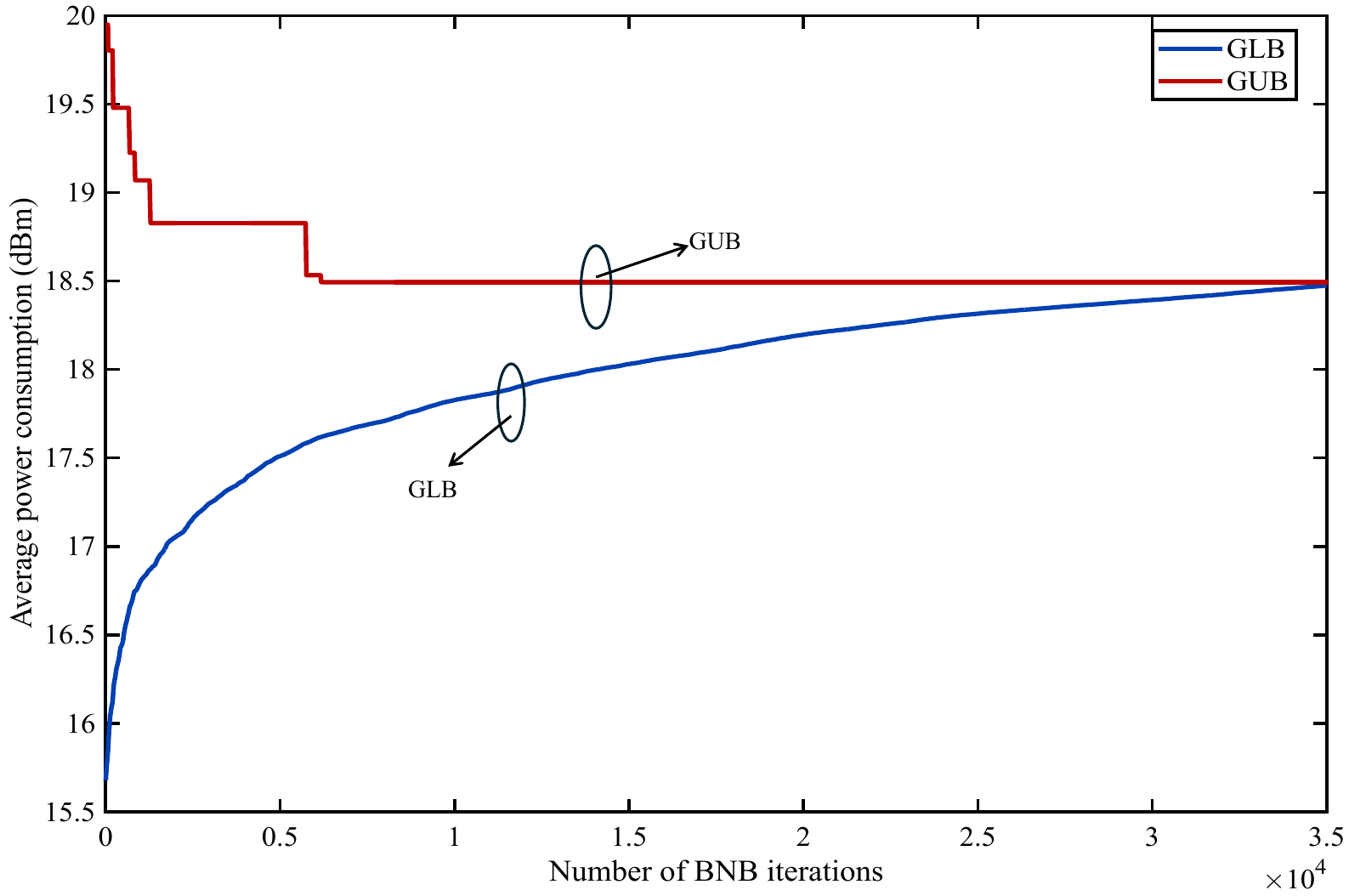}
  \caption{\label{figure:convergence}Convergence of BnB for the MWSU scenario.
  } 
   \vspace{-10pt} 
\end{figure}

\begin{figure}[t]
  \centering
  \includegraphics[width=0.45\textwidth]{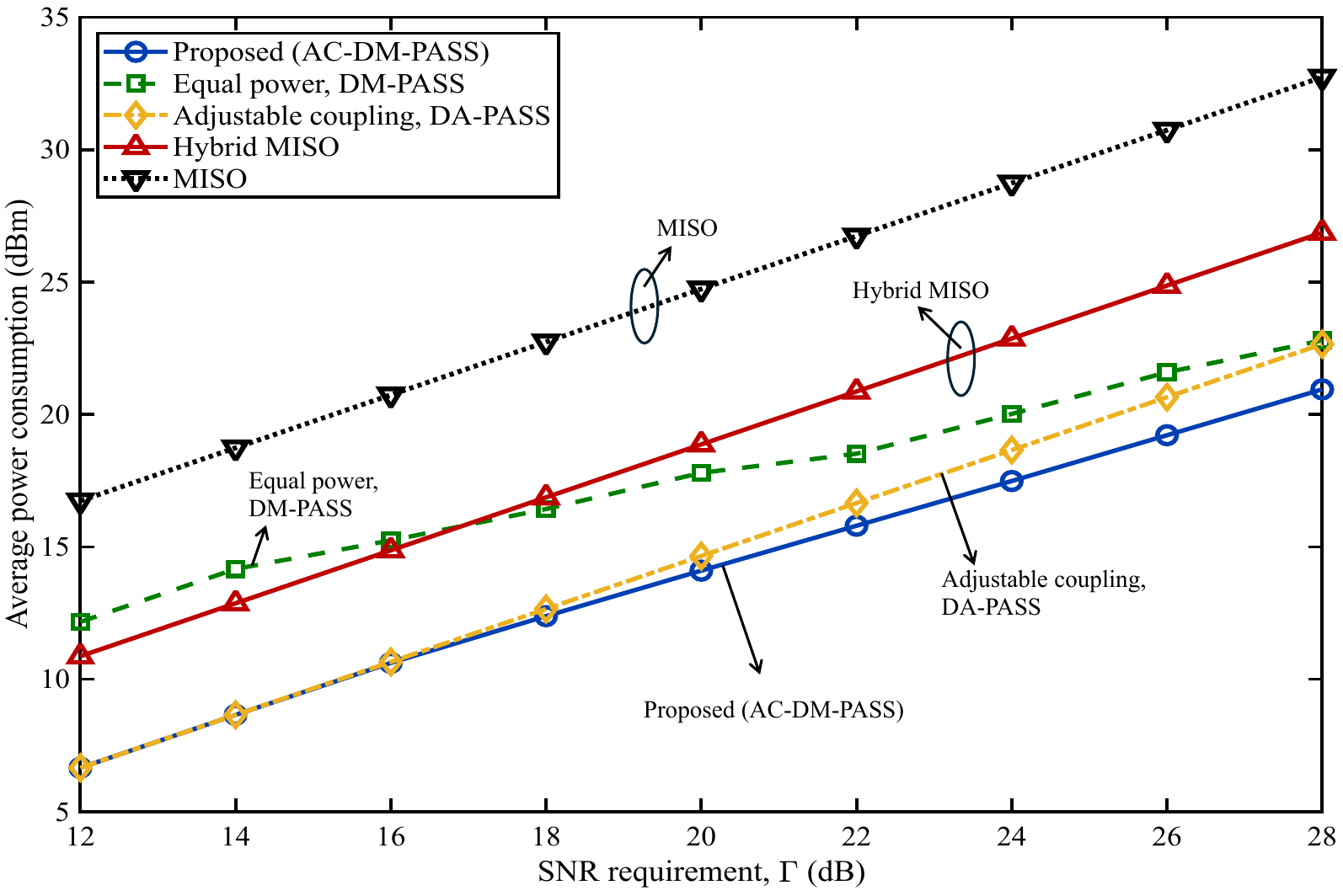}
  \caption{\label{figure:bnbsnr} Performance comparisons in the MWSU scenario.
  } 
   \vspace{-10pt} 
\end{figure}

Fig.~\ref{figure:convergence} demonstrates the convergence behavior of the proposed globally optimal BnB algorithm for the MWSU scenario. It can be observed that GLB increases monotonically as the branching process proceeds, while GUB gradually decreases. Moreover, the gap between the GUB and GLB continuously shrinks and eventually converges within the prescribed tolerance, indicating that the proposed BnB algorithm can guarantee convergence to the globally optimal solution within a finite number of iterations. Compared to exhaustive search, which requires evaluating approximately $5.31\times10^{17}$ feasible configurations in the considered setup, the proposed BnB algorithm converges after approximately $3.5\times10^4$ iterations, demonstrating the effectiveness of the proposed bound-pruning strategy in significantly reducing the search complexity. Nevertheless, the computational complexity of the globally optimal BnB algorithm still increases rapidly with the number of PAs, discrete candidate positions, and  coupling spacing levels. Therefore, for the MWMU scenario, a low-complexity GA-PSO algorithm is further proposed to efficiently handle the highly coupled joint optimization of PA positions, coupling strength, and transmit beamforming.

Fig.~\ref{figure:bnbsnr} compares the system performance of different architectures under different SNR requirements $\Gamma$ in the MWSU scenario. The required power consumption increases monotonically with $\Gamma$ for all schemes, while the proposed PASS consistently achieves the lowest power consumption. Compared with the equal-power DM-PASS scheme, the proposed design significantly reduces the power consumption by adaptively controlling the radiation power of different PAs, which avoids inefficient uniform power radiation. Compared with the adjustable coupling DA-PASS, the proposed scheme reduces the required power consumption, especially in the high-SNR regime, demonstrating the importance of PA position optimization under stringent QoS requirements. Moreover, both conventional MISO and hybrid MISO require substantially higher power consumption due to their fixed antenna structures and limited capability of reconfiguring large-scale path loss.

\vspace{-1.05em}
\subsection{Multi-User Scenario}

Fig.~\ref{figure:sinr} illustrates the average power consumption versus the minimum SINR requirement $\Gamma_k$ for different architectures. It can be observed that the required power consumption increases monotonically with $\Gamma_k$ for all schemes, since higher SINR requirements demand stronger signal enhancement and interference suppression. Among all compared schemes, the proposed PASS consistently achieves the lowest power consumption. Compared with equal power DM-PASS, the proposed design achieves a significant power reduction of up to $42.2\%$, demonstrating that adjustable radiation power control provides additional DoFs for more efficient power allocation among PAs. The performance gain is more pronounced in the low-SINR regime, where flexible radiation power allocation plays a dominant role in reducing the required transmit power. Moreover, compared with the adjustable coupling DA-PASS scheme, the proposed scheme reduces the power consumption by up to $39.7\%$, which verifies the importance of practical PA position optimization in reconfiguring the large-scale path loss and propagation phases. The corresponding gain becomes more evident at high SINR requirements, since PA position optimization provides additional spatial DoFs for reconfiguring the large-scale path loss and propagation phases, while also reducing the inter-user channel correlation. In contrast, the proposed PASS architecture achieves up to $99\%$ power reduction compared with conventional MIMO and hybrid MIMO schemes, highlighting the remarkable capability of PASS in mitigating large-scale path loss and enhancing spatial beamforming flexibility through reconfigurable pinching beamforming.

\begin{figure}[t]
  \centering
  \includegraphics[width=0.45\textwidth]{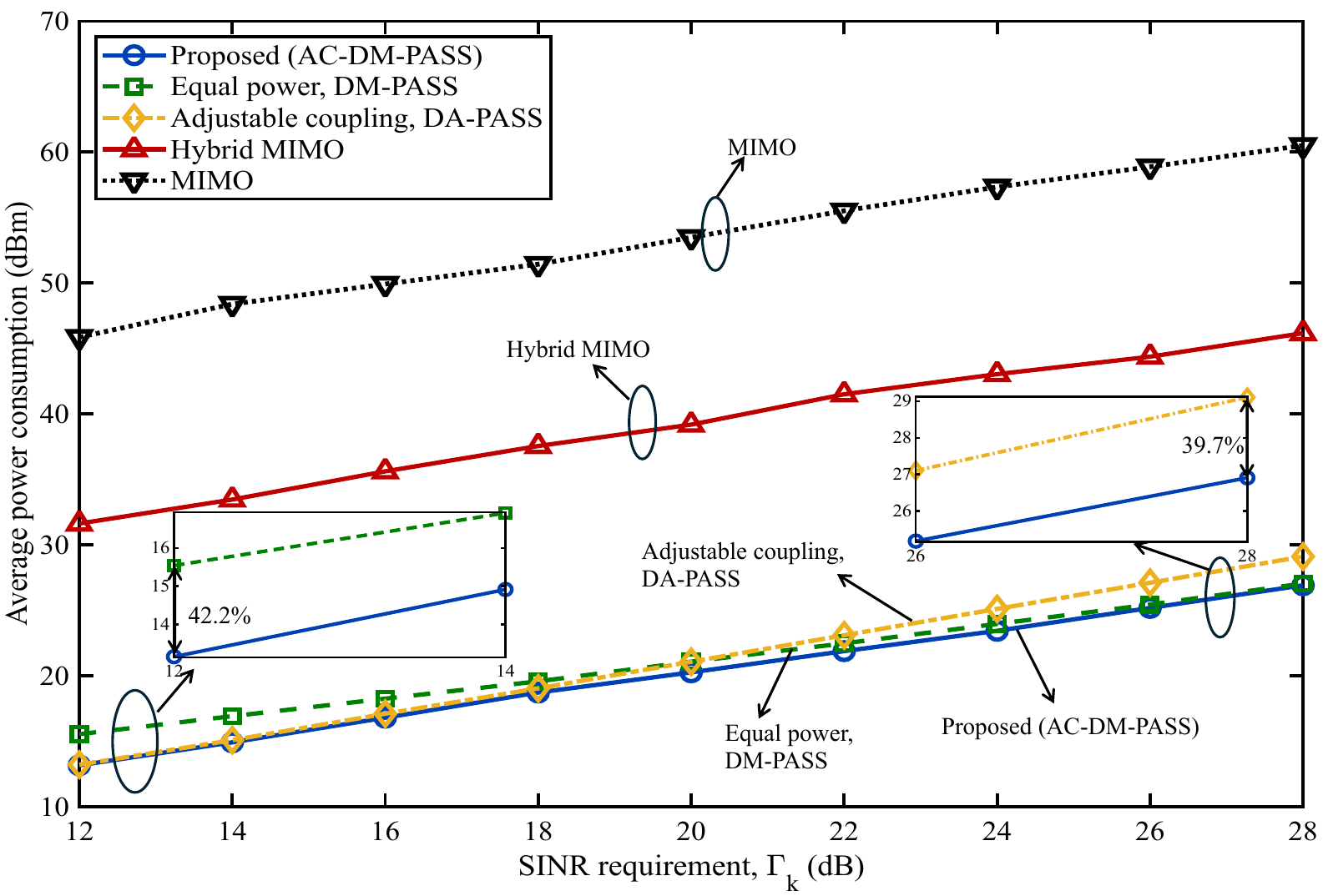}
  \caption{\label{figure:sinr}Average power consumption versus SINR requirement.
  } \vspace{-10pt} 
\end{figure}

\begin{figure}[t]
  \centering
  \includegraphics[width=0.45\textwidth]{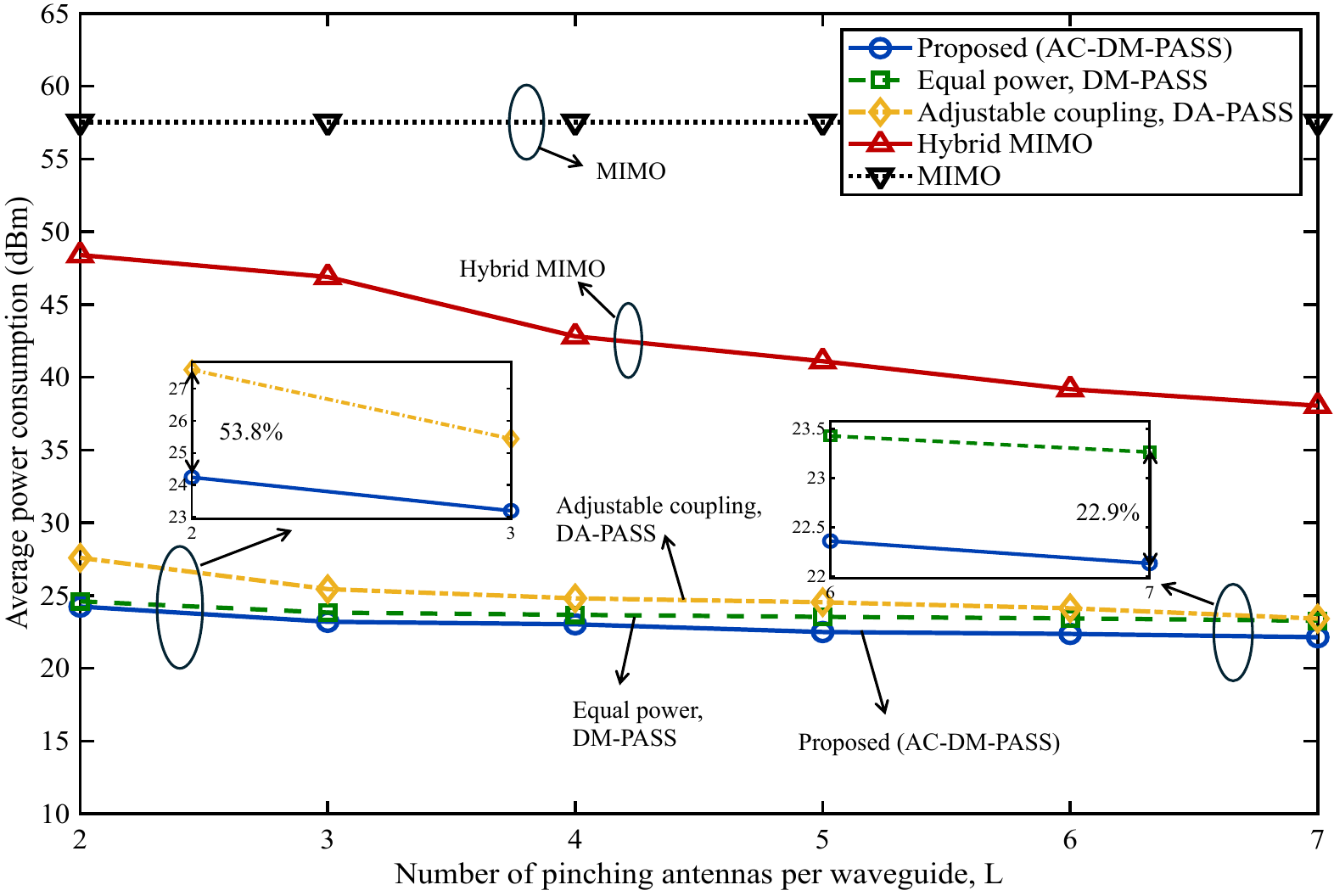}
  \caption{\label{figure:L}Average power consumption  versus number of PAs per waveguide.
  } \vspace{-10pt} 
\end{figure}

\begin{figure}[t]
  \centering
  \includegraphics[width=0.45\textwidth]{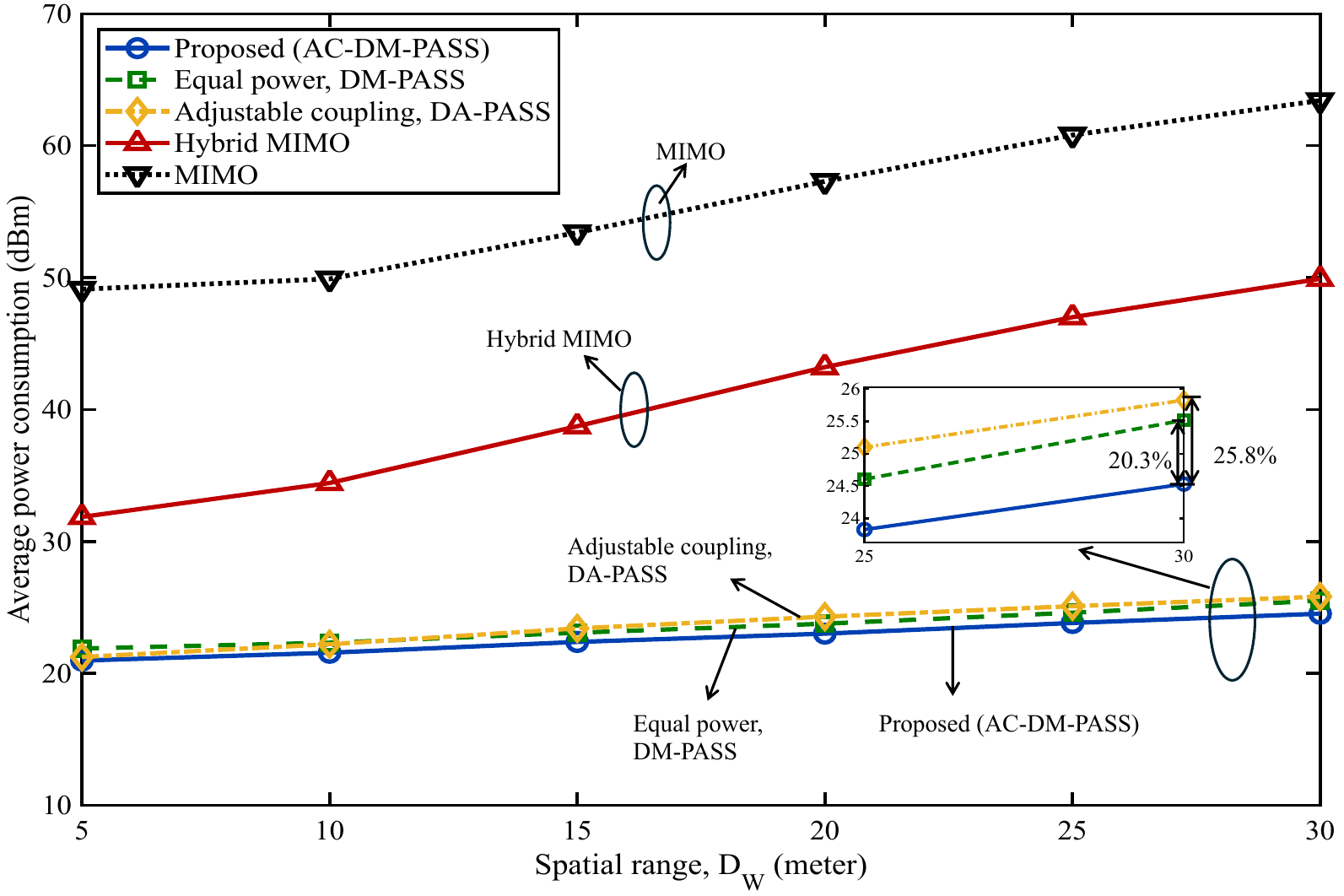}
  \caption{\label{figure:SpatialRange}Average power consumption versus spatial range.
  } \vspace{-10pt} 
\end{figure}

\begin{figure}[t]
  \centering
  \includegraphics[width=0.45\textwidth]{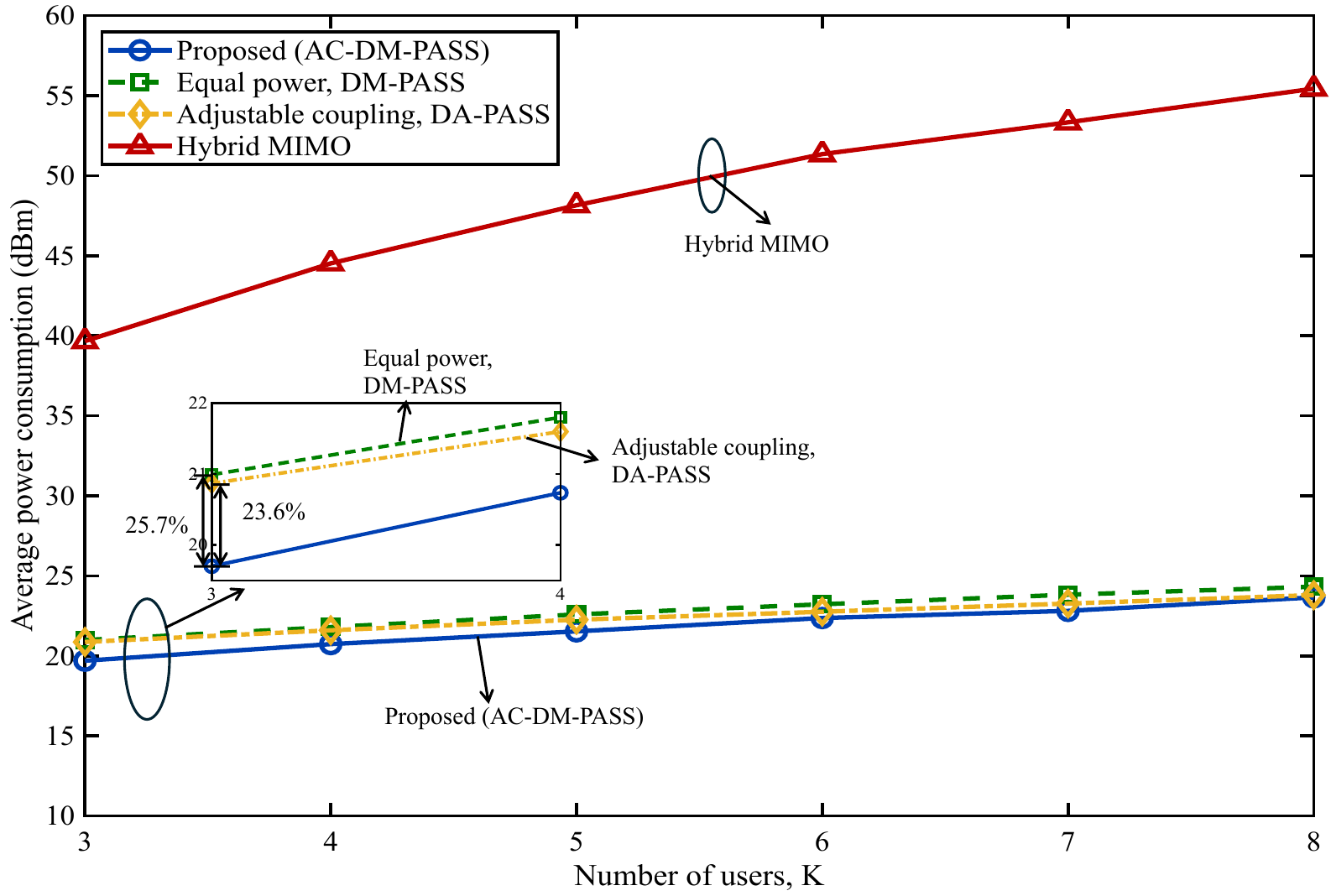}
  \caption{\label{figure:userK}Average power consumption versus  number of users.
  } \vspace{-10pt} 
\end{figure}

\begin{figure}[t]
  \centering
  \includegraphics[width=0.45\textwidth]{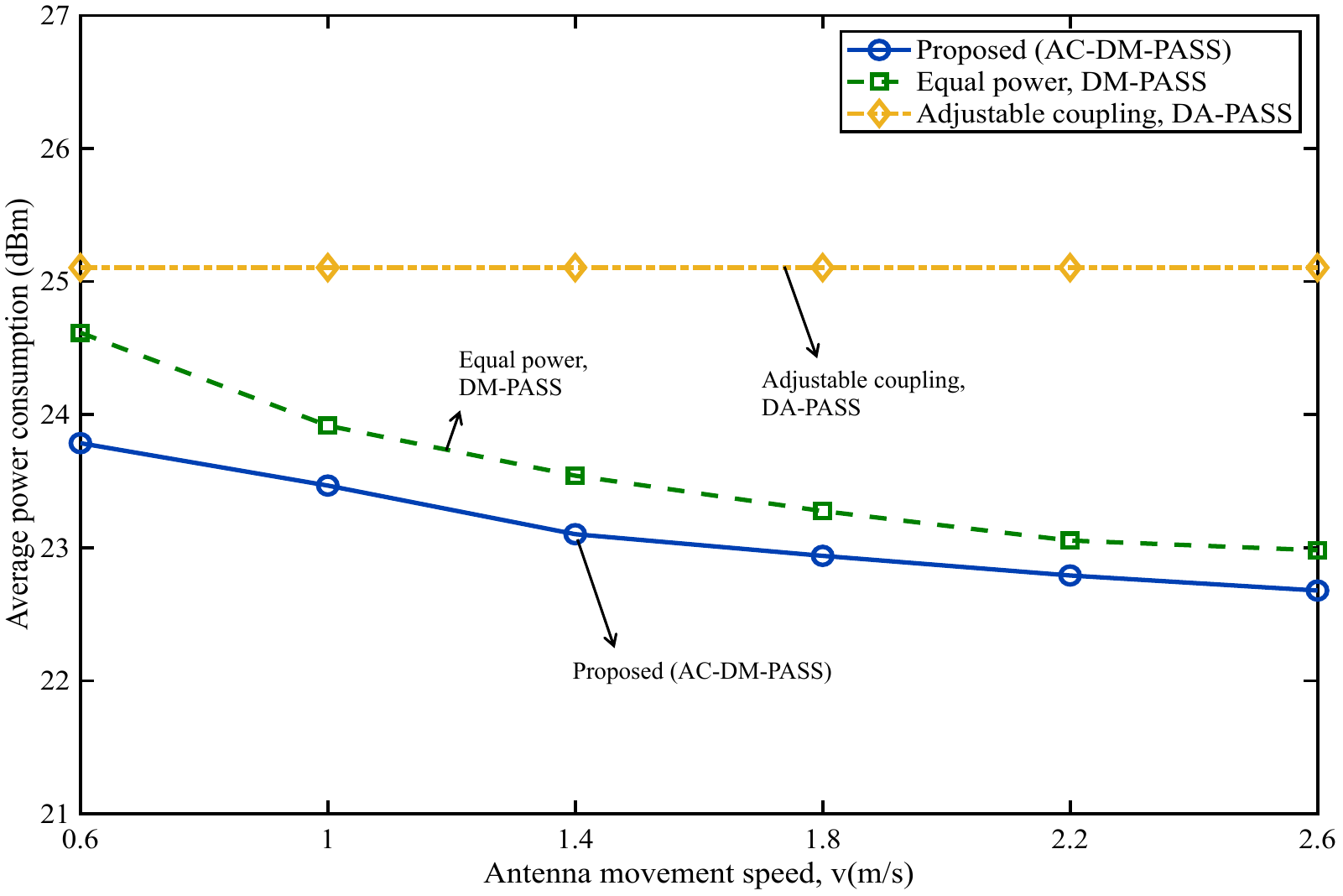}
  \caption{\label{figure:vmove}Average power consumption versus antenna movement speed.
  }
\end{figure}

Fig.~\ref{figure:L} illustrates the impact of the number of PAs per waveguide $L$. The average power consumption decreases with increasing $L$ for all PASS-based schemes, since more PAs provide denser radiation points and greater spatial flexibility for pinching beamforming and LoS link construction. Nevertheless, the performance gain gradually saturates as $L$ increases, since the additional spatial gain achieved by increasing the number of PAs becomes marginal. Moreover, the proposed scheme consistently achieves the lowest power consumption among all considered schemes. Compared with the adjustable coupling DA-PASS, the proposed AC-DM-PASS reduces the average power consumption by up to $53.78\%$, demonstrating that PA position optimization becomes increasingly important when the number of PAs is limited. In contrast, conventional MIMO and hybrid MIMO still require substantially higher power consumption. By dynamically adjusting radiation points close to users and flexibly adjusting the spacing levels for radiation power control, the proposed PASS framework significantly improves channel conditions and reduces the required transmit power.

Fig.~\ref{figure:SpatialRange} illustrates the average power consumption versus the spatial range $D^{\mathrm{W}}$ for different architectures. It can be observed that the required power consumption increases with $D^{\mathrm{W}}$ for all schemes due to the increased large-scale path loss in wider coverage areas. In contrast, the proposed PASS framework consistently achieves the lowest power consumption. This is because the proposed scheme jointly optimizes the PA positions, coupling spacing levels, and transmit beamforming, enabling the PAs to move closer to users and establish stronger LoS links. Compared with the equal power DM-PASS and adjustable coupling DA-PASS baselines, the proposed scheme achieves approximately $20.3\%$ and $25.8\%$ power reduction at $D^{\mathrm{W}}=30$ m, respectively. Moreover, the performance gain becomes more significant as the spatial range increases, which demonstrates the superiority of practical PA position optimization and adjustable radiation power control in wide-area deployment scenarios.

Fig.~\ref{figure:userK} illustrates the system performance of different architectures versus the number of users $K$ with  $\Gamma_k = 20$ dB. It can be observed that the average power consumption increases with $K$ for all considered schemes, since serving more users introduces stronger inter-user interference and requires additional transmit power to satisfy all users' SINR requirements. Nevertheless, the proposed PASS framework consistently achieves the lowest power consumption among all compared schemes, demonstrating the effectiveness of practical PA position optimization and adjustable radiation power control in reducing inter-user channel correlation and improving spatial beamforming flexibility. Compared with the equal power baseline and the fixed-position baseline, the proposed scheme achieves approximately $25.7\%$ and $23.6\%$ power reduction at $K=3$, respectively.  In contrast, conventional MIMO and hybrid MIMO architectures suffer from substantially higher power consumption due to their fixed antenna structures and limited capability in reconfiguring the wireless propagation environment.

Fig.~\ref{figure:vmove} illustrates the impact of the antenna movement speed $v$ on the average power consumption of different PASS-based transmission schemes. It can be observed that the average power consumption decreases as $v$ increases. This is because a larger movement speed enlarges the feasible movement region of each PA, thereby improving spatial flexibility for PA deployment and reducing the required transmit power. In addition, the motion power consumption is inversely proportional to $v$, such that increasing $v$ directly reduces the movement power cost. Moreover, the proposed scheme consistently achieves the lowest power consumption among all considered schemes. In contrast, adjustable coupling DA-PASS scheme remains unchanged with increasing $v$ since the PA positions are fixed and no movement power consumption is involved. Nevertheless, higher movement speeds generally require more sophisticated hardware and higher implementation costs, which introduces a tradeoff between hardware complexity and system performance.

\section{Conclusion} \label{sec:conclusion}

This paper proposed a practical PASS framework that enables discrete radiation power control and localized discrete antenna movement. The coupling strength adjustment was realized by tuning the coupling spacing between each PA and the waveguide, which provides additional DoFs for flexible radiation power control. Moreover, the position of each PA is selected from a set of discrete mounting points within a localized region determined by the PA movement speed and movement duration. By incorporating practical in-waveguide attenuation and antenna motion power consumption, the PA positions, coupling strength, and transmit beamforming were jointly optimized to minimize the total power consumption subject to users' SINR requirements and PAs' localized motion constraints. A globally optimal BnB algorithm was proposed for the MWSU scenario, with theoretical guarantees on convergence and optimality. To reduce the computational complexity, we further proposed a GA-PSO algorithm for the MWMU scenario, where GA-assisted PSO optimized the discrete variables and SOCP reconstructed the transmit beamforming. Simulation results demonstrated that the proposed design significantly reduced the average power consumption compared with existing PASS schemes and conventional MIMO architectures.

\appendix

\subsection{Proof of Theorem \ref{Theorem:BnB}}\label{proof:Theorem}

Let $\mathcal F$ denote the feasible set of problem $\mathcal P_{\mathrm S}'$, i.e.,
\begin{equation}
\mathcal F
\triangleq
\left\{
(\mathbf x,\mathbf s)\,\middle|\,
x_{l,n}\in\mathcal X_{l,n}^{\mathrm{local}},
s_{l,n}\in\mathcal S,\,
\eqref{cons:P0_motion},\,\eqref{cons:P0_min_spacing}
\right\}.
\end{equation}
Since each $\mathcal X_{l,n}^{\mathrm{local}}$ and $\mathcal S$ is finite, $\mathcal F$ is also finite. The globally optimal objective value of $\mathcal P_{\mathrm S}'$ is denoted as
\begin{equation}
f^\star
\triangleq
\min_{(\mathbf x,\mathbf s)\in\mathcal F}
P_{\mathrm{total}}^{\mathrm S}(\mathbf x,\mathbf s).
\end{equation}
For any BnB node $\mathcal B$, let $\mathcal F(\mathcal B)\subseteq\mathcal F$ denote the subset of feasible configurations contained in $\mathcal B$, and define the local optimum over node $\mathcal B$ as
\begin{equation}
f^\star(\mathcal B)
\triangleq
\min_{(\mathbf x,\mathbf s)\in\mathcal F(\mathcal B)}
P_{\mathrm{total}}^{\mathrm S}(\mathbf x,\mathbf s).
\end{equation}

We first show the validity of the lower bound. For any $(\mathbf x,\mathbf s)\in\mathcal F(\mathcal B)$, the construction of $C^{\mathrm{UB}}(\mathcal B)$  is given by
$\|\mathbf c(\mathbf x,\mathbf s)\|_2^2 \le C^{\mathrm{UB}}(\mathcal B)$.
Moreover, by the definition of $d_{l,n}^{\min}(\mathcal B)$, we have
$\left|x_{l,n}-x_{l,n}^{\mathrm{prev}}\right| \ge d_{l,n}^{\min}(\mathcal B),\,\forall l,n$.
Based on Eq. (\ref{eq:Emn_def}) and Eq. (\ref{eq:Ptotal_def}), the objective follows that
\begin{align}
P_{\mathrm{total}}^{\mathrm S}(\mathbf x,\mathbf s) \ge
\frac{T_2}{T}
\frac{\Gamma\sigma^2}
{C^{\mathrm{UB}}(\mathcal B)}
+
\frac{P_{\mathrm m}}{vT}
\sum_{n=1}^{N}\sum_{l=1}^{L}
d_{l,n}^{\min}(\mathcal B)
= f_{\mathrm{LB}}(\mathcal B).
\end{align}
Therefore, the lower bound satisfies
$f_{\mathrm{LB}}(\mathcal B)\le f^\star(\mathcal B),\,\forall \mathcal B$,
which proves that $f_{\mathrm{LB}}(\mathcal B)$ is a valid local lower bound.

We then show the validity of the upper bound. The upper bound $f_{\mathrm{UB}}(\mathcal B)$ is computed from an explicitly feasible solution $(\widehat{\mathbf x},\widehat{\mathbf s})\in\mathcal F(\mathcal B)$ as
$f_{\mathrm{UB}}(\mathcal B) = P_{\mathrm{total}}^{\mathrm S}(\widehat{\mathbf x},\widehat{\mathbf s})$.
Thus, the upper bound satisfies
$f_{\mathrm{UB}}(\mathcal B)\ge f^\star(\mathcal B)$,
and the incumbent solution associated with $\mathrm{GUB}$ is always feasible for $\mathcal P_{\mathrm S}'$.

Let $\mathcal A$ denote the active-node set. The global lower bound is updated as
$\mathrm{GLB} = \min_{\mathcal B\in\mathcal A} f_{\mathrm{LB}}(\mathcal B)$,
while $\mathrm{GUB}$ is the best feasible objective value found so far. Since the active nodes together with the incumbent solution cover all non-pruned candidates that may improve the current solution, and since $f_{\mathrm{LB}}(\mathcal B)$ is valid for every active node, we have $\mathrm{GLB}\le f^\star\le \mathrm{GUB}$.

We next verify the correctness of the pruning operation. If a node $\mathcal B$ satisfies
$f_{\mathrm{LB}}(\mathcal B)\ge \mathrm{GUB}$,
then for any feasible configuration $(\mathbf x,\mathbf s)\in\mathcal F(\mathcal B)$, we have
\begin{equation}
P_{\mathrm{total}}^{\mathrm S}(\mathbf x,\mathbf s)
\ge
f_{\mathrm{LB}}(\mathcal B)
\ge
\mathrm{GUB}.
\end{equation}
Therefore, no feasible solution contained in $\mathcal B$ can improve the current incumbent, and discarding $\mathcal B$ does not affect the global optimality of the BnB search.

Then, the branching operation partitions the remaining candidate domains $\mathcal X_{l,n}(\mathcal B)$ and $\mathcal S_{l,n}(\mathcal B)$ into smaller non-overlapping subsets. Since all position and coupling spacing domains are finite, repeated branching can generate at most $\prod_{n\in\mathcal N}\prod_{l\in\mathcal L} |\mathcal X_{l,n}^{\mathrm{local}}| |\mathcal S|$ singleton configurations in the worst case. Therefore, the BnB tree contains a finite number of leaf nodes. Since each iteration either prunes a node or branches it into smaller finite subsets, the proposed BnB algorithm must terminate after a finite number of iterations.

At termination, let $(\mathbf x_\epsilon,\mathbf s_\epsilon)$ be the incumbent solution returned by the algorithm, with objective value
$f_\epsilon = P_{\mathrm{total}}^{\mathrm S}(\mathbf x_\epsilon,\mathbf s_\epsilon) =\mathrm{GUB}$.
Then, we have
\begin{equation}
0\le f_\epsilon-f^\star
\le
\mathrm{GUB}-\mathrm{GLB}.
\end{equation}
When the stopping criterion $\frac{\mathrm{GUB}-\mathrm{GLB}} {\max\{1,|\mathrm{GUB}|\}} \le \epsilon$ is satisfied, it follows that 
$f_\epsilon-f^\star \le \epsilon\max\{1,|f_\epsilon|\}$.
Therefore, $(\mathbf x_\epsilon,\mathbf s_\epsilon)$ is an $\epsilon$-optimal solution of problem $\mathcal P_{\mathrm S}'$ in the relative-gap sense. The corresponding transmit beamforming vector $\mathbf w_\epsilon$ can be recovered by the MRT solution, and thus $(\mathbf x_\epsilon,\mathbf s_\epsilon,\mathbf w_\epsilon)$ gives an $\epsilon$-optimal solution to the problem $\mathcal{P}_{\mathrm{S}}$. This completes the proof.

\bibliographystyle{IEEEtran}
\bibliography{citepassone}

\vfill

\end{document}